%% file: second_revision.tex
\documentclass[journal,10pt,twocolumn]{IEEEtran}
\usepackage{float}
\usepackage{amsmath,amsthm, amsfonts}
\input{notation}

\usepackage[ruled,vlined]{algorithm2e}
\usepackage{graphicx}

\usepackage{array}
\usepackage{bm} 
\usepackage{booktabs}      % for \toprule, \midrule, \bottomrule
\usepackage{siunitx}       % for consistent units (optional but recommended)

\usepackage[most]{tcolorbox}
\tcbuselibrary{skins, breakable}
\newtcolorbox{ieeeproblem}[2][]{%
  breakable, enhanced, colback=white, colframe=black,
  title={#2}, fonttitle=\bfseries, % etc.
  #1
}

\usepackage{textcomp}
\usepackage{stfloats}
\usepackage{url}
\usepackage{verbatim}
\usepackage{tikz}
\usetikzlibrary{arrows.meta, calc, decorations.pathreplacing, positioning, 3d}

\usepackage[noadjust]{cite}
\usepackage{xcolor}
\usepackage{amsmath}
\allowdisplaybreaks[2]          % amsmath may break align at reasonable points

\usepackage{amssymb}
\usepackage{mathtools}
\usepackage{epstopdf}
\usepackage{graphicx}    % ensures 
\usepackage[
    font=small,
    skip=4pt,
    labelfont=bf,
    belowskip=4pt,
    compatibility=false
]{caption}
\usepackage{subcaption}
\usepackage{lipsum}

\usepackage{cuted}

\newtheorem{proposition}{Proposition}

\begin{document}
% \bstctlcite{BSTcontrol}
\title{{Optimal Anchor Placement for Wireless Localization in Mixed LOS and NLOS Scenarios}\\}

\author{
~Gaurav~Duggal,~R.~Michael~Buehrer,~Harpreet~S.~Dhillon, Jeffrey~H.~Reed
 \thanks{\\
 Gaurav Duggal, R. Michael Buehrer, Harpreet S. Dhillon, and Jeffrey H. Reed are with Wireless@VT, Bradley Department of Electrical and Computer Engineering, Virginia Tech,  Blacksburg, VA, 24061, USA. Email: \{gduggal, rbuehrer, hdhillon, reedjh\}@vt.edu.\\ The support of NIST PSCR PSIAP through grant 70NANB22H070, NSF through grant CNS-1923807, CNS-2107276 and NIJ graduate research fellowship through grant 15PNIJ-23-GG-01949-RES is gratefully acknowledged.\\ }
}

% \markboth{Indoor Positioning by Exploiting Diffraction}%


\maketitle

\begin{abstract}
This paper develops a general optimal experimental design framework for anchor placement in localization and related estimation problems whose Fisher information matrix (FIM) can be expressed as a weighted sum of rank-one matrices. Within this framework, anchor locations are selected to optimize Cramér--Rao lower bound (CRLB)-based performance measures, while accounting jointly for measurement geometry and information strength. As a challenging motivating application, we consider localization in diffraction-dominated non-line-of-sight (NLOS) outdoor-to-indoor (O2I) environments, where the direct path is severely attenuated or blocked. Although diffraction-based time-of-arrival (TOA) measurements can provide useful positioning information, optimal anchor placement is less intuitive than in conventional line-of-sight (LOS) settings because both the measurement directions and their information strengths depend on the propagation environment.
The proposed $S_n,r_n$ -based FIM representation compactly captures this interaction between geometry and signal strength. For a single target, it yields a polygon-closure interpretation and characterizes the conditions under which A-, D-, and E-optimal anchor geometries coincide. We then extend the framework to region-wide anchor placement by constructing a precomputed anchor--target dictionary and formulating min--max E- and D-optimal anchor-selection problems as mixed-integer second-order cone programs. Numerical results demonstrate the resulting anchor layouts, localization bounds, and computational tradeoffs in a mixed-LOS/NLOS O2I environment. While the analysis is instantiated and evaluated for diffraction-based O2I localization, its underlying results apply broadly to estimation and anchor-placement problems with weighted rank-one-sum FIMs.
\end{abstract}


\section{Introduction}
TOA–based localization~\cite{zekavat2019handbook} forms the foundation of many modern positioning systems, ranging from global navigation satellite system (GNSS) to indoor and industrial networks that rely on ultra-wideband (UWB), Wi-Fi, and cellular 5G/6G technologies. Its widespread adoption stems from its ability to provide high-precision range estimates using simple timing measurements between transmitters and receivers. For this study we use {\em anchors} to  refer to entities whose positions are known {\em a priori}, such as base stations, access points, or reference beacons, while {\em targets} represent the agents whose positions are to be estimated. The placement of these anchors has an impact on localization accuracy and robustness. The localization error depends on both the accuracy of the TOA measurements and the spatial configuration of the anchors relative to the target. While TOA measurement error is primarily determined by factors such as signal bandwidth, signal strength, and noise level, the anchor geometry dictates how these measurement errors propagate into the estimate of the position. This geometric sensitivity is commonly quantified by the geometric dilution of precision (GDOP), which captures the amplification of measurement error due to anchor placement. However, GDOP does not account for signal coverage, which affects the measurement error and is hence directly influenced by the anchor–target distances, transmit power, and propagation conditions that govern the signal-to-noise-ratio (SNR) of each TOA link. Hence, localization accuracy can be enhanced by deploying anchors in a manner that
(a) improves the SNR between each anchor and the target, and
(b) ensures a favorable spatial configuration that yields a good anchor geometry (low GDOP). 
\par
Classical works on optimal anchor deployment have predominantly considered 
LOS scenarios, where localization accuracy is determined by 
the anchor geometry and the SNR of the direct propagation paths. 
Under these conditions, anchors that are spatially well-separated and exhibit 
large angular diversity with respect to the target position or the region of interest \cite{ho2008sensor,nguyen2016optimal,mengoptimal2013,meng2016optimal,sadeghi2020optimal,sahu2022optimal,xu2021optimal,wu2023optimization,aubrytsp2023,fatima2024optimal,rao2024iterative,xu2025optimal,zhang2025optimal} 
are regarded as geometrically advantageous, as they help minimize position ambiguity and improve estimation accuracy.
 Under this notion of geometric optimality, and assuming distance-independent path loss the theoretical lower bound on the root mean squared error (RMSE) for any unbiased position estimator has also been derived \cite{xu2019optimal,xu2021optimal,xu2025optimal}. However, this notion of geometric optimality is valid only under LOS propagation. In realistic  mixed LOS/NLOS environments, propagation effects such as reflection, diffraction \cite{duggal2025diffractionaidedwirelesspositioning, pallaprolu2025}, and attenuation distort range measurements and degrade SNR, and also invalidate the LOS-based geometric intuition for localization applications. Furthermore, most existing formulations define optimality with respect to a single fixed target whose position is to be estimated. These approaches typically assume prior knowledge of the target position, which is not very practical since in real scenarios the target location is unknown and may vary across a wide region. It is therefore essential to extend anchor deployment strategies to the multi-target setting, enabling area-wide optimization and ensuring consistent localization performance across the entire region of interest. Preliminary efforts toward multi-target formulations exist, notably \cite{xu2021optimal}, which analyzes a two-target case, and \cite{wu2023optimization,aubrytsp2023}, which extends the framework to multiple targets. However, these studies assume distance-independent path loss and pure LOS conditions. 
\par
The optimality of anchor geometry is closely related to the field of 
\textit{optimal experimental design}~\cite{pukelsheim2006optimal}, 
where scalar functions of the CRLB matrix, 
derived from the FIM, are optimized to 
achieve desired localization accuracy. There are three basic criteria: {\em A-optimality} which minimizes the trace of the CRLB matrix \cite{ho2008sensor,sadeghi2020optimal,xu2021optimal,sahu2022optimal,wu2023optimization,aubrytsp2023,fatima2024optimal,zhang2025optimal,xu2025optimal}, {\em D-optimality} \cite{mengoptimal2013,sadeghi2020optimal,sahu2022optimal,aubrytsp2023} which minimizes the determinant of the CRLB matrix and {\em E-optimality} \cite{sadeghi2020optimal,sahu2022optimal,fatima2024optimal} which minimizes the largest eigenvalue of the CRLB matrix. Optimization frameworks that optimize anchor geometry over all three criteria have been proposed, in particular  
\cite{sadeghi2020optimal} shows all three criteria to be equivalent for their formulation that assumes distance-independent path loss, while \cite{sahu2022optimal} extends this to include distance-dependent path loss.
\par
In terms of optimization techniques, most optimality formulations are nonconvex with respect to the anchor positions, mainly due to the presence of distance-dependent path loss. Some tackle this non-convexity using heuristic approaches like genetic algorithms \cite{nguyen2016optimal}. Other approaches in literature directly derive analytical solutions \cite{meng2016optimal,xu2019optimal,sadeghi2020optimal}. Other works propose frameworks which tackle the nonconvex problem iteratively by updating one block of variables at a time using convex surrogate functions that locally approximate and upper bound the original nonconvex objective, ensuring a monotonic decrease in the objective value at each iteration \cite{aubrytsp2023,wu2023optimization,fatima2024optimal,xu2025optimal}. In \cite{rao2024iterative}, the authors employ a heuristics-based iterative approach where anchors are added sequentially as a heuristic in the interest of tractability. Although the authors of \cite{joshi2008sensor} do not discuss anchor optimality, they propose an interesting approach to solving the  problem of choosing a set of $K$ sensor measurements, from a set of $M$ possible or potential sensor measurements, that minimizes the error in estimating some parameter of the measurement. This is in general a nonconvex continuous problem which is reduced to a convex combinatorial sensor-selection problem. We extend their ideas to formulate our own optimization framework in this study.

In addition to theoretical interest, this problem has critical relevance to indoor localization for next-generation emergency response networks, such as firefighting operations, disaster rescue, and active-shooter incidents \cite{duggal2023icc}. In such scenarios, precise localization of both first responders and at-risk individuals is vital for situational awareness and safety. Relying on pre-installed localization infrastructure (e.g., Wi-Fi access points or cellular base stations) is often impractical or unreliable during emergencies, when power failures, infrastructure damage, or dense building materials can lead to degraded signal coverage indoors. To address this, mobile anchors—such as UAVs — can be rapidly deployed around the affected building to provide ad hoc localization support \cite{harishion2023}. These systems operate in O2I environments where signal propagation is dominated by NLOS mechanisms, particularly diffraction from window edges rather than direct LOS paths \cite{duggal2025diffractionaidedwirelesspositioning}. In general,  diffraction has been shown to emerge due to bending around edges \cite{zhang2022time,pallaprolu2022wiffract,pallaprolu2025}. Designing optimal anchor deployments in such environments thus requires jointly accounting for anchor geometry, placement flexibility, and the physical signal propagation mechanisms in NLOS scenarios.
Motivated by these limitations, the main contributions of this paper are as follows:
\begin{itemize}
    \item \textbf{Unified mixed LOS/NLOS FIM framework and general \((S_n,r_n)\) reformulation:}
    We develop a unified Fisher-information framework for TOA-based localization in mixed LOS/NLOS environments with diffraction-dominated O2I propagation. The framework combines a generalized path-length model with a continuous path-loss/SNR model that remains well behaved across the LOS/NLOS transition. This leads to a compact \((S_n,r_n)\)-based representation of the two-dimensional (2D) localization FIM. Although derived here for the proposed diffraction model, the reformulation is more general and applies to any 2D localization modality whose FIM can be written as a sum of weighted rank-one outer products of unit vectors.

    \item \textbf{Single-target optimality theory via the \((S_n,r_n)\) reformulation:}
    Building on the \((S_n,r_n)\) representation, we analyze A-, D-, and E-optimal anchor placement for a single target under distance-independent ranging-information weights. We show that all three criteria are simultaneously optimized when the residual term vanishes, yielding a polygon-closure interpretation in the complex plane and necessary and sufficient conditions for optimality through the generalized polygon inequality. Thus, in this regime, A-, D-, and E-optimal designs coincide.

    \item \textbf{Region-wide multi-target anchor optimization over a discretized anchor-target dictionary:}
    We extend the \((S_n,r_n)\)-based framework to region-wide optimization over multiple targets while reintroducing distance-dependent path loss. By discretizing the feasible three-dimensional (3D) anchor region, we recast the continuous placement problem as a combinatorial anchor-selection problem over a precomputed anchor-target dictionary. This yields two exact MISOCP formulations for the discretized min--max problem, corresponding to the E- and D-optimal criteria, while A-optimality is handled through provable relationships showing that these objectives also control the A-optimal criterion. Moreover, when solved by branch-and-bound, the formulations admit standard feasibility, \(\epsilon\)-optimality, and infeasibility certificates for the discretized problem.
\end{itemize}

\section{Mathematical Preliminaries}

\subsection{Fisher information for Parameter Estimation}
Let $\bm{y}\triangleq \{y_1,\cdots,y_M\}$ be the observation vector, $\bm{\eta}\triangleq\{\eta_1,\cdots,\eta_K\}$ be the parameter vector of the distribution of $\bm{y}$ and the likelihood function be $\chi(\bm{y}\mid\bm{\eta})$.
\begin{ndef}
\label{definition_FIM}
The FIM \cite{kay1993fundamentals} $\bm{\mathcal{I}}_{\bm{\eta}}\in\mathbb{R}^{K\times K}$ is
\begin{equation}
\bm{\mathcal{I}}_{\bm{\eta}}
\;=\; \mathbb{E}_{\bm{y}\mid\bm{\eta}}\!\left[\nabla_{\bm{\eta}}\ln\chi(\bm{y}\mid\bm{\eta})\,\nabla_{\bm{\eta}}\ln\chi(\bm{y}\mid\bm{\eta})^{\!\top}\right],
\end{equation}

\end{ndef}

% The rank of the FIM provides critical insight into the identifiability and estimability of the parameters in a statistical model.
% \begin{remark}
% \label{remark_rank_FIM}
% If the Fisher information Matrix (FIM) is full rank, each parameter is locally identifiable and can be independently estimated from the data \cite{kay1993fundamentals}. Conversely, if the FIM is rank deficient, the variance of some parameter estimates becomes unbounded, rendering those parameters unidentifiable.
% \end{remark}

\subsection{FIM Under Reparameterization}
In many estimation problems, the likelihood is naturally expressed in terms of one parameter vector, while the quantities of interest are represented by another parameterization related through a known differentiable mapping. In such cases, the Fisher information for the parameters of interest can be obtained by reparameterizing the Fisher information of the original parameter vector as follows.
\begin{ndef}
\label{def_FIM_transformation}
Let $\bm{\eta}\in\mathbb{R}^K$ be a parameter vector with FIM
$\bm{\mathcal{I}_{\bm{\eta}}}\in\mathbb{R}^{K\times K}$. Let $\bm{\xi}\in\mathbb{R}^L$
be another parameter vector such that
\[
\bm{\eta}=\bm{f}(\bm{\xi}),
\qquad
\bm{f}:\mathbb{R}^L\rightarrow\mathbb{R}^K,
\]
where $\bm{f}$ is differentiable. Define the Jacobian matrix
\[
\bm{J}(\bm{\xi}) \triangleq \frac{\partial \bm{\eta}}{\partial \bm{\xi}}
\in\mathbb{R}^{K\times L}.
\]
Then the FIM of $\bm{\xi}$,
$\bm{\mathcal{I}_{\bm{\xi}}}\in\mathbb{R}^{L\times L}$, is given by
\begin{equation}
    \bm{\mathcal{I}_{\bm{\xi}}}
    =
    \bm{J}(\bm{\xi})^T
    \bm{\mathcal{I}_{\bm{\eta}}}
    \bm{J}(\bm{\xi}).
\end{equation}
\end{ndef}

\subsection{Equivalent FIM}
In many applications, only a subset of parameters is of interest, with the rest treated as nuisance variables. The equivalent FIM (EFIM) compactly captures the information pertaining to the parameters of interest. 

\begin{ndef}
\label{def_EFIM}
Consider the parameter vector $\bm{\eta}=[\bm{\eta}_1^T, \bm{\eta}_2^T]^T$ with 
$\bm{\eta}\in\mathbb{R}^N$ and $\bm{\eta}_1 \in \mathbb{R}^K$, where $K<N$. 
The overall FIM can be partitioned as
\[
\mathbf{I}_{\bm{\eta}}=
\begin{bmatrix}
\bm{\mathcal{I}_{\bm{\eta}_1,\bm{\eta}_1}} & 
\bm{\mathcal{I}_{\bm{\eta}_2,\bm{\eta}_1}^T} \\
\bm{\mathcal{I}_{\bm{\eta}_2,\bm{\eta}_1}} & 
\bm{\mathcal{I}_{\bm{\eta}_2,\bm{\eta}_2}}
\end{bmatrix},
\]
where 
$\bm{\mathcal{I}_{\bm{\eta}_1,\bm{\eta}_1}} \in \mathbb{R}^{K\times K}$, 
$\bm{\mathcal{I}_{\bm{\eta}_2,\bm{\eta}_1}} \in \mathbb{R}^{(N-K)\times K}$, and 
$\bm{\mathcal{I}_{\bm{\eta}_2,\bm{\eta}_2}} \in \mathbb{R}^{(N-K)\times(N-K)}$. 
The EFIM of $\bm{\eta}_1$, obtained by treating $\bm{\eta}_2$ as nuisance parameters, 
is given by the Schur complement \cite{horn2012matrix} as 
\begin{equation}
\label{eq_definition_EFIM}
\bm{\mathcal{I}_{\bm{\eta}_1,\bm{\eta}_1}^{[e]}}
= \bm{\mathcal{I}_{\bm{\eta}_1,\bm{\eta}_1}} 
- \bm{\mathcal{I}_{\bm{\eta}_2,\bm{\eta}_1}^T}
\bm{\mathcal{I}_{\bm{\eta}_2,\bm{\eta}_2}^{-1}}
\bm{\mathcal{I}_{\bm{\eta}_2,\bm{\eta}_1}}.
\end{equation}
\end{ndef}

\subsection{CRLB on Parameter Estimation}
The CRLB provides a fundamental lower bound on the covariance of any unbiased estimator.
\begin{ndef}
\label{def_crlb}
 Let $\hat{\bm{\eta}} \in \mathbb{R}^K$ be an unbiased estimator of the parameter vector $\bm{\eta} \in \mathbb{R}^K$, derived from observations $\bm{y}$. Then, the error covariance matrix of $\hat{\bm{\eta}}$ satisfies the information inequality:
\begin{equation}
    \mathbb{E}_{\bm{y}|\bm{\eta}} \left[ (\hat{\bm{\eta}} - \bm{\eta})(\hat{\bm{\eta}} - \bm{\eta})^\top \right] \succeq \bm{\mathcal{I}_{\bm{\eta}}^{-1}}
\end{equation}
where $\bm{\mathcal{I}_{\bm{\eta}}} \in \mathbb{R}^{K \times K}$ is the FIM, defined element-wise as in Definition \ref{definition_FIM} and $\bm{\mathcal{I}_{\bm{\eta}}^{-1}}$ is called the CRLB matrix.
\end{ndef}

\section{Model for Localization in Mixed LOS/NLOS O2I Scenario}
% \begin{figure*}[!htbp]
%     \begin{subfigure}{0.45\linewidth}
%         \centering
%         % \includegraphics[clip, trim=4.5cm 8cm 5.5cm 8.5cm, width=1\linewidth]
% \includegraphics[clip, trim=0cm 0cm 0cm 0cm, width=1\linewidth]{figs/System_Model.pdf}
%         \caption{System Model}
%         \label{fig_system_model_exact}
%     \end{subfigure}
%     \begin{subfigure}{0.45\linewidth}
%         \centering
% \includegraphics[clip, trim=0cm 0cm 0cm 0cm, width=1\linewidth]{figs/System_Model_approx.pdf}
%         \caption{System Model Approx}
%         \label{fig_system_model_approx}
%     \end{subfigure}

% \caption{\GD{Solve weird svg to pdf issue (causing missing text and misaligned boxes)} } 
%     \label{fig_system_model}
% \end{figure*}

\begin{figure}[htbp]
    \centering
    \includegraphics[width=1\linewidth]{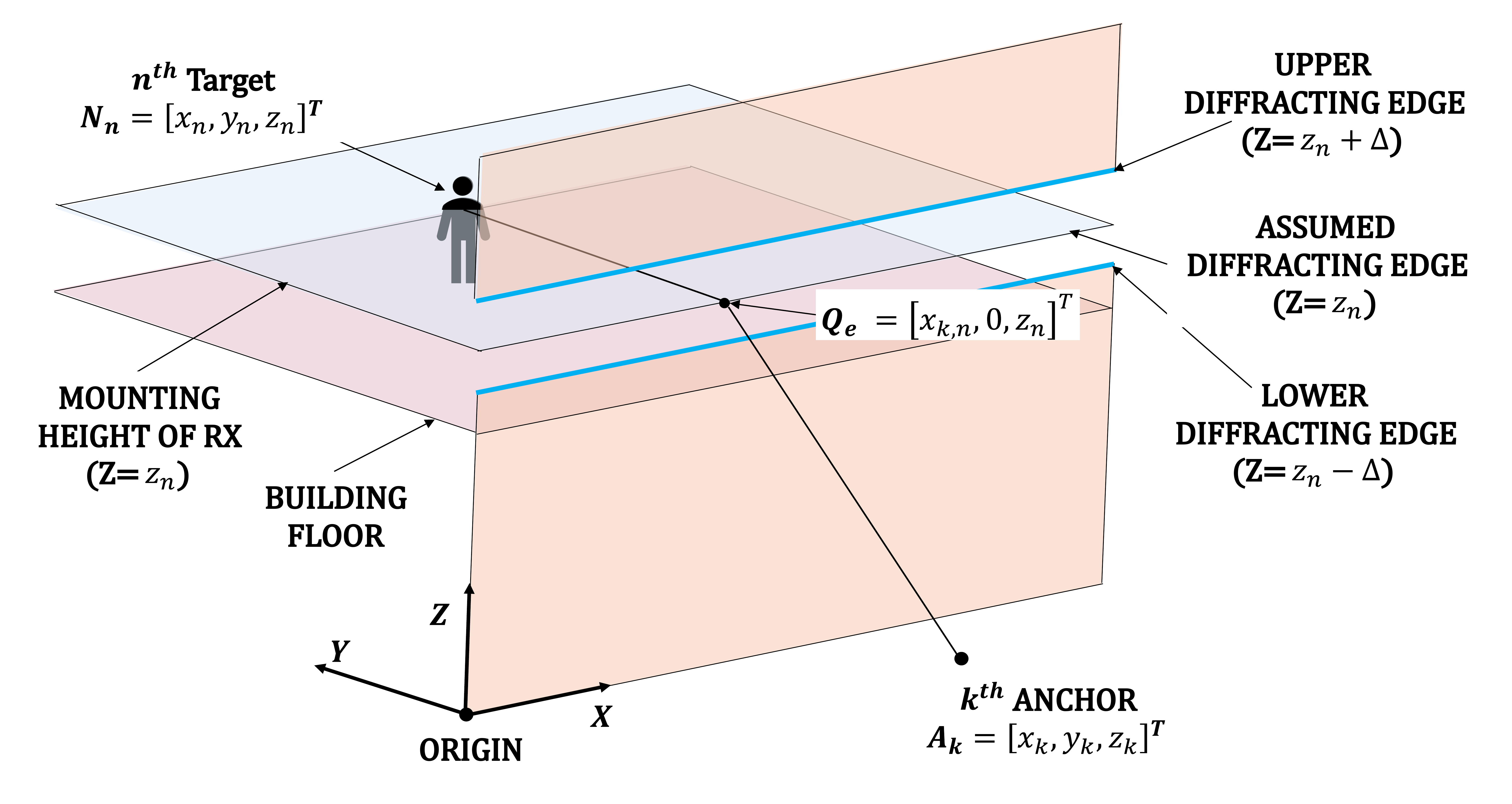}
    \caption{In the O2I scenario, we have \(K\) anchors transmitting orthogonal signals which are received by the $n^{\text{th}}$ target inside the building. 
The received signal includes several MPCs, from which the ranging measurement corresponding to the diffraction path length 
\(\bm{A_k}\bm{Q}_e\bm{N}_n\) is extracted. 
A bandwidth of \(200\) MHz is assumed to provide fine delay resolution.
In this study, we assume that the relevant diffraction
path is reliably identifiable as the first-arriving path and is
sufficiently separated from nearby MPCs so that the single-path
ranging-information approximation in \eqref{eq_ranging_uncertainty_diffraction_path} is valid.}

    \label{fig_system_model}
\end{figure}

% \begin{figure}[h]
%     \centering
%    \includegraphics[width=0.8\linewidth]{figs/System_Model.pdf}
%   \caption{System Model}
%   \label{fig_system_model}
% \end{figure}

As shown in \figref{fig_system_model}, we consider an O2I scenario with $K$ anchors and $N$ indoor targets, indexed by $\mathcal{K}\triangleq\{1,\dots,K\}$ and $\mathcal{N}\triangleq\{1,\dots,N\}$, respectively. Each anchor $k\in\mathcal{K}$ transmits a signal that is received by target $n\in\mathcal{N}$. The anchor waveforms are mutually orthogonal, so the targets experience no inter-anchor interference. From the received signal of each anchor, a target obtains one Time-of-Flight (TOF) range estimate, typically using the first-arrival path criterion \cite{duggal2025}. In the O2I setting considered here, this criterion isolates the diffraction path, whose path-length model, as we show in Section \ref{section_diffraction_path_length_model}, is continuous across the LOS/NLOS transition region.

Consistent with \cite{duggal20243d}, we assume that the floor containing target \(n\) is known, e.g., from barometer data \cite{Firstnetroadmap}. The localization problem for that target therefore reduces to estimating its 2D in-floor position from \(K\) TOF measurements. However, the anchor geometry remains fully three-dimensional: anchors are not restricted to the plane \(Z=z_n\), and may be located anywhere in 3D space. Hence, even when targets are distributed across multiple floors, the corresponding anchor--target path lengths are determined by the full 3D geometry and generally differ across both anchors and floors. In this sense, the known-floor assumption reduces only the dimension of the target state to be estimated, not the dimension of the anchor deployment or the underlying propagation model. For the corresponding full 3D localization formulation, see \cite{duggal2025diffractionaidedwirelesspositioning}. None of these prior works addresses the optimal anchor placement problem considered here.

\subsection{Received Signal Model}
Assume that each anchor transmits a waveform $s_k(t)$, such that it is mutually orthogonal between the anchors and has a flat power spectral density (PSD) i.e.
\begin{equation}
\label{eq_flat_PSD}
|S_k(f)|^2 =
\begin{cases}
\frac{1}{B}, & \text{for } f \in \left[-\frac{B}{2}, \frac{B}{2}\right] \\
0, & \text{otherwise}
\end{cases}.
\end{equation}
The signal received at target $n$ from anchor $k$ during the observation interval $T_{\text{ob}}$ can be expressed as
\begin{align}
\label{eq_received_signal}
\begin{split}
r_n(t) &= \sum_{l=1}^{L}h_l\; s_k(t-\tau_l)+ n_n(t),\quad t\in (0,T_{\text{ob}}) \\
&= \underbrace{h_1   s_k(t-\tau_1)}_{\text{Diffraction MPC}} + \underbrace{\sum_{l=2}^{L}h_l   s_k(t-\tau_l)}_{\text{Other MPCs}}+ n_n(t)
\end{split}
\end{align}
This equation is to be interpreted as the received signal being a superposition of delayed replicas of the transmitted waveform from the $k^{\text{th}}$ anchor corresponding to $L$ multipath components (MPCs) represented by index set $\mathcal{L}\triangleq \{1,\dots,L\}$. Each MPC $l \in \mathcal{L}$ is characterized by a delay $\tau_l$, determined by its propagation path length, and a path gain $h_l$, which accounts for attenuation along that path. MPCs generally arise from electromagnetic interactions with the environment, including specular reflections from planar surfaces (e.g., walls, floors, and ceilings), transmissions through the building exterior and inter-floor structures, and diffractions from window edges present on each floor. 
At FR2 and FR3 in O2I NLOS settings, the direct outdoor-to-indoor path is effectively absent due to heavy attenuation from exterior walls, floors, and ceilings, making diffraction the dominant mechanism; its MPCs can be used for localization \cite{duggal2025diffractionaidedwirelesspositioning}. Among NLOS MPCs (reflections and diffractions), diffraction paths are reliably isolated via the first-arriving path (FAP). Frequency-dependent results in \cite{duggal2025,duggal2025diffractionaidedwirelesspositioning} show FR3 incurs lower diffraction loss than FR2, enabling better localization performance.

Let $l=1$ correspond to the diffraction MPC between the $k^{\text{th}}$ anchor and the target $n$, as illustrated in \figref{fig_system_model}. The TOF-based ranging measurement is a noisy estimate of the length of the diffraction path.
In the next section, we begin by introducing the diffraction path length formulation from \cite{duggal2025diffractionaidedwirelesspositioning}. In this study we obtain a simplified version that shows that it is in fact a generalization of the conventional Euclidean path model.

\subsection{A Unified LOS/NLOS Path Model}
\label{section_diffraction_path_length_model}
Since the propagation delays for the various MPCs are determined by their associated path lengths, we adopt a diffraction-based path length model first introduced in ~\cite{Ang1999DiffractionPoints} and then independently rediscovered and explicitly used for localization in ~\cite{duggal20243d,duggal2025diffractionaidedwirelesspositioning}. We first show both these formulations are mathematically equivalent. The full derivation is in Appendix~\ref{appendix_diffraction_path_length_derivation}. This model provides a unified treatment of LOS and NLOS, with a smooth transition from Euclidean to diffraction path length. It eliminates the need to label paths as LOS or NLOS—the geometry determines the appropriate regime.

In \figref{fig_system_model}, two diffraction edges run parallel to the \(x\)-axis, representing the upper and lower window edges on the target’s floor. They lie at the intersections of the plane \(Y=0\) with \(Z=z_n+\Delta\) (upper edge) and \(Z=z_n-\Delta\) (lower edge). Here $2\Delta$ is the vertical size of the windows. Let $e \in \{l,u\}$ represent the lower and upper edge respectively. The location of the diffraction point on the respective diffracting edge for the $k^{\mathrm{th}}$ anchor, is given by $\bm{Q}_{k,n}^{\langle\text{e}\rangle} = [x_{k,n}^{\langle\text{e}\rangle},0,z_n+s_e\Delta]^T$, where 
\begin{align}
x_{k,n}^{\langle\text{e}\rangle}  = x_n\!-\!(x_n\!-\!x_k) \frac{\sqrt{y_n^2\!+\!\Delta^2}}{\sqrt{y_k^2\!+\!(z_n\!+\!s_e\Delta-z_k)^2}\!+\!\sqrt{y_n^2\!+\!\Delta^2}}\\ \nonumber
\qquad \qquad \qquad e\in\{u,l\},\; s_u=+1,\; s_l=-1.
\end{align}    
Now, the diffraction path length $p_{k,n}$ between the $k^{\mathrm{th}}$ anchor at $\bm{A}_k = [x_k,y_k,z_k]^T$ and $n^{\text{th}}$ target at $[x_n,y_n,z_n]^T$ can then be expressed as  
{\small
\begin{align}
\label{eq_diffraction_path_length}
p_{k,n}^{\langle\text{e}\rangle}\!=\!\sqrt{(x_k\!-\!x_n)^2\!+\!\left( \sqrt{y_k^2\!+\!(z_n\!+\!s_e\Delta\!-\!z_k)^2}\!+\! \sqrt{y_n^2\!+\!\Delta^2} \right)^2} \\ \nonumber
\qquad \qquad \qquad e\in\{u,l\},\; s_u=+1,\; s_l=-1.
\end{align}    
}
We assume we can reliably isolate the shortest diffraction path using the first-arriving path principle \cite{duggal2025diffractionaidedwirelesspositioning}. If the anchor is above the target, the upper-edge diffraction path is shortest; if below, the lower-edge path is shortest. We can simplify the model by assuming $\Delta=0$.
 Note that we made the same assumption in our previous work~\cite{duggal2025diffractionaidedwirelesspositioning} where we demonstrated that $\Delta=0$ has minimal effect on the positioning accuracy. Hence, this leads to the following simplified path model
\begin{align}
\label{eq_diffraction_path_length_approx}
    p_{k,n} \approx \sqrt{(x_k - x_n)^2 + \left( \sqrt{y_k^2 + (z_k-z_n)^2} + y_n \right)^2}.
\end{align}
\begin{remark}
    When the anchor lies in the same horizontal plane as the target, i.e., substitute $z_k = z_n$ in \eqref{eq_diffraction_path_length_approx}, the expression simplifies to 
    \[
        p_{k,n} = \sqrt{(x_k - x_n)^2 + (|y_k| + y_n)^2}.
    \]
    In this case, the additional diffraction-related detour vanishes, and the path reduces to the direct Euclidean distance in the 2D plane $Z=z_n$. This corresponds precisely to a LOS path between the anchor and the target, hence presenting the diffraction path length as a generalization of the Euclidean path model.
\end{remark}
\subsection{Unified Path Loss for O2I Mixed LOS/NLOS Scenario}

There are several ways to model the SNR of a diffraction path \cite{duggal2025diffractionaidedwirelesspositioning,duggal2025,mcnamara1990introduction,balanis2012advanced}. In this
work, we use the ITU-R knife-edge diffraction (KED) model as an
approximate, easy-to-compute power-domain model that does not require
ray tracing, phase, or polarization information \cite{ITU_diffraction}.
This is suitable for our mixed LOS/NLOS O2I setting, where same-floor
anchor-target links are treated as LOS, while links from anchors above
or below the target floor are modeled as diffraction-dominated NLOS
links \cite{duggal2023icc}.

For the simplified O2I geometry, define the effective diffraction
intrusion as
\[
h\triangleq |z_k-z_n|,
\]
which is zero for the same-floor LOS case and positive for anchors
located above or below the target floor. This is an effective
nonnegative intrusion parameter for the simplified O2I KED model, rather
than the signed clearance used in the canonical ITU-R terrain-profile
model. The distances from the anchor and target to the diffraction point are
\begin{align}
\label{eq_3d_distance_diff_point_fresnel_params}
d_{1}
=
\left\lVert \bm{A}_k-\bm{Q}_{k,n}^{\langle e\rangle}\right\rVert,
\qquad
d_{2}
=
\left\lVert \bm{N}_n-\bm{Q}_{k,n}^{\langle e\rangle}\right\rVert .
\end{align}
The Fresnel-Kirchhoff parameter is
\[
\nu
=
h
\sqrt{
\frac{2f_c(d_1+d_2)}{c d_1d_2}
},
\]
where \(c\) is the speed of light and \(f_c\) is the carrier frequency.
For NLOS links, the KED excess loss is \cite{ITU_diffraction}
\begin{align}
\label{eq_excess_path_loss}
L_d(\nu)
=
6.9+20\log_{10}
\Big(
\sqrt{(\nu-0.1)^2+1}+\nu-0.1
\Big)
\ \text{dB}.
\end{align}
For same-floor LOS links, we set \(L_d(\nu)=0\).

The received SNR in dB is then
\begin{align}
\label{eq_SNR_diffraction_path}
\mathrm{SNR}_{\mathrm{dB}}
&=
P_t^{\mathrm{dBW}}
+G_t+G_r
-\mathrm{FSPL}(R,f_c)
-L_d(\nu) \\
&\qquad\qquad\qquad-\Bigl[
10\log_{10}(K_BT B)+\mathrm{NF}
\Bigr],\nonumber
\end{align}
where
\[
\mathrm{FSPL}(R,f_c)
=
20\log_{10}\left(\frac{4\pi Rf_c}{c}\right),
\qquad
R=\lVert \bm{N}_n-\bm{A}_k\rVert .
\]
Here \(P_t\) is the transmit power, \(G_t\) and \(G_r\) are the antenna
gains, \(K_BTB\) is the thermal noise power, and \(\mathrm{NF}\) is the
receiver noise figure.

\textit{Note:} The Friis term in \eqref{eq_SNR_diffraction_path} uses
the anchor-target distance \(R\) as the free-space reference, while
\(L_d(\nu)\) models the ITU-R KED excess loss. Replacing \(R\) by the
diffraction length \(p_{k,n}=d_1+d_2\) would additionally account for
spreading along the diffracted path; this secondary correction is
neglected here.

\subsection{TOA Ranging Estimation in Multipath Scenarios}

% Our goal in this section is to obtain the CRLB on the TOF-based ranging estimate corresponding to the diffraction MPC in the presence of other MPCs. 
\begin{lemma}
Assuming the receive signal model in \eqref{eq_received_signal}, and conditioning on the complex path gains ${h_l}_{l=1}^{L}$, the FIM for estimating our delay vector $\bm{\tau}\triangleq[\tau_1,\dots,\tau_L]^T$is
% \begin{align}
% \small
% \label{eq_FIM_tau}
% &[\bm{J_\eta}]_{l_1,l_2} = \\ \nonumber
% &\begin{cases}
% \displaystyle
% \frac{2\pi^2 B^2 |h_{l_1}|^2}{3N_0}, &  l_1 = l_2 \\[10pt]
% \frac{2}{N_0} \operatorname{Re} \left\{
% h_{l_1}^* h_{l_2} 
% \frac{8\pi \left( \pi \delta_{l_1l_2} \cos(\pi B \delta_{l_1l_2}) - 2 \sin(\pi B \delta_{l_1l_2}) \right)}{\delta_{l_1l_2}^3} \right\}, & l_1 \ne l_2
% \end{cases}\\ \nonumber
% & \quad \quad \quad \quad \quad \quad \quad \quad \quad \quad \quad \quad \quad \quad \quad \quad \quad \quad \quad  l_1,l_2 \in 1\dots L
% \end{align}
{\small
\begin{align}
\label{eq_FIM_tau}
&[\bm{\mathcal{I}_\bm{\tau}}]_{l_1,l_2} = \\[-2pt] \nonumber
&\begin{cases}
\displaystyle
\frac{2\pi^2 B^2 |h_{l_1}|^2}{3N_0}, 
& l_1 = l_2 \\[12pt]
\displaystyle
\frac{2}{N_0}\,\Re\!\Bigg\{
h_{l_1}^* h_{l_2}\;
\frac{1}{(\pi B)\,\delta_{l_1 l_2}^{3}}
\Big[
(\pi B\delta_{l_1 l_2})^{2}\sin\!\big(\pi B\delta_{l_1 l_2}\big) \\[4pt]
\qquad\qquad\qquad\qquad\qquad\quad
+\,2(\pi B\delta_{l_1 l_2})\cos\!\big(\pi B\delta_{l_1 l_2}\big) \\[4pt]
\qquad\qquad\qquad\qquad\qquad\quad
-\,2\sin\!\big(\pi B\delta_{l_1 l_2}\big)
\Big]
\Bigg\},
& l_1 \ne l_2
\end{cases} \\[2pt] \nonumber
& \hspace{6cm} l_1,l_2 \in \{1,\dots,L\}
\end{align}
}

Here $l_1,l_2$ are pairwise indices of the MPCs and correspond to the row and column indices of the FIM and $\delta_{l_1,l_2}= \tau_{l_1}-\tau_{l_2}$ is called the path overlap coefficient.    
\end{lemma}

\begin{proof}
Refer to Appendix \ref{section_FIM_first_arriving_path} for proof.
\end{proof}
It should be observed that the off-diagonal elements of the FIM in \eqref{eq_FIM_tau} depend on the path overlap parameter $\delta_{l_1,l_2}, l_1,l_2\in \mathcal{L}$, whereas the diagonal elements remain independent of it. For the diagonal terms, the ratio of the path gain to the noise power can be interpreted as the SNR of the $l^{\text{th}}$ path. We can thus express the $l^{\text{th}}$ diagonal entry as  
\[
    [\bm{\mathcal{I}_{\bm{\tau}}}]_{l,l} = \frac{2\pi^2B^2\text{SNR}_l}{3}, \quad l \in \{1,\dots,L\}.
\]

Next, we partition the delay parameters as $\bm{\tau}=[\bm{\tau}_1^T,\bm{\tau}_2^T]^T$ where $\bm{\tau}_1=\tau_1$ is the delay of the diffraction path which we are interested in and $\bm{\tau}_2=[\tau_2,\dots,\tau_L]^T$ as the delays of the other $L-1$ MPCs which are considered as nuisance parameters. The corresponding FIM can be partitioned as
\[
\bm{\mathcal{I}_{\bm{\tau}}} =
\begin{bmatrix}
\bm{\mathcal{I}_{\tau_1, \tau_1}} & \bm{\mathcal{I}_{\tau_{2:L},\tau_1}^T} \\
\bm{\mathcal{I}_{\tau_{2:L},\tau_1}} & \bm{\mathcal{I}_{\tau_{2:L}, \tau_{2:L}}}
\end{bmatrix}
\]
where 
\[
\bm{\mathcal{I}_{\tau_1,\tau_1}} \in \mathbb{R}, \bm{\mathcal{I}_{\tau_1,\tau_{2:L}}} \in \mathbb{R}^{1\times L-1}, \bm{\mathcal{I}_{\tau_{2:L},\tau_{2:L}}} \in \mathbb{R}^{(L-1) \times (L-1)}.
\]
The EFIM $\bm{\mathcal{I}_{\tau_1}^{[e]}}$ corresponding to the delay of the diffraction path can now be expressed, using Definition~\ref{def_EFIM}, as
\[
\bm{\mathcal{I}_{\tau_1,\tau_1}^{[e]}} =
\bm{\mathcal{I}_{\tau_1,\tau_1}} -  \bm{\mathcal{I}_{\tau_2,\tau_L}^{[l]}}
\]
where,
\[
\bm{\mathcal{I}_{\tau_2,\tau_L}^{[l]}} =
\bm{\mathcal{I}_{\tau_1,\tau_{2:L}}} \bm{\mathcal{I}_{\tau_{2:L},\tau_{2:L}}^{-1}} \bm{\mathcal{I}_{\tau_1,\tau_{2:L}}^T}
\]
and this represents the loss in information due to overlapping paths.
As the delay separation between MPCs $l_1$ and $l_2$ increases, i.e., as $\delta_{l_1,l_2} \to \infty$, the CRLB obtained using Definition~\ref{def_crlb} for estimating the range of the diffraction path converges to the corresponding single-path bound. This occurs because the information loss due to path overlap vanishes in the large-separation regime. Thus, when the MPCs are effectively resolvable, the variance of the ranging estimate for the diffraction path between the
$k^{\text{th}}$ anchor and the $n^{\text{th}}$ target satisfies the CRLB
\begin{align}
\label{eq_ranging_uncertainty_diffraction_path}
\operatorname{var}(\hat{r}_{k,n})
\;\ge\; \frac{1}{\lambda_{k,n}},
\qquad
\lambda_{k,n} \;=\; \frac{2\pi^{2} B^{2}\,\mathrm{SNR}_{k,n}}{3c^{2}}.
\end{align}
Here, $\lambda_{k,n}$ is the \emph{ranging information weight}. Empirically, evaluating the CRLB of the RMSE of the ranging error versus path overlap at
$50$, $100$, and $200\mathrm{MHz}$ in \figref{fig_crlb_tof_estimation} shows that a delay separation larger than
$1/B$ is typically sufficient to make the overlap-induced information loss negligible; i.e.,
\[
\lim_{\delta_{l_1,l_2} \gg \tfrac{1}{B}} \bm{\mathcal{I}_{\tau_2,\tau_L}^{[l]}} = 0.
\]

% As the MPCs $l_1$ and $l_2$ become well separated in delay i.e. ,$\delta_{l_1,l_2} \to \infty$, the CRLB (obtained using Definition~\ref{def_crlb}) for estimating the range of the diffraction path converges to that of the single-path case. Thus when path overlap is negligible, the variance of the ranging estimate for the diffraction path between the
% $k^{\text{th}}$ anchor and the $n^{\text{th}}$ target satisfies the CRLB
% \begin{align}
% \label{eq_ranging_uncertainty_diffraction_path}
% \operatorname{var}(\hat{r}_{k,n})
% \;\ge\; \frac{1}{\lambda_{k,n}},
% \qquad
% \lambda_{k,n} \;=\; \frac{2\pi^{2} B^{2}\,\mathrm{SNR}_{k,n}}{3c^{2}}.
% \end{align}
% Here, $\lambda_{k,n}$ is the \emph{ranging information weight}. As the multipath
% components become well separated, the information loss due to path overlap
% vanishes. Empirically, evaluating the CRLB of the RMSE of the ranging error versus path overlap at
% $50$, $100$, and $200\,\mathrm{MHz}$ in \figref{fig_crlb_tof_estimation} shows that an overlap larger than
% $1/B$ is typically sufficient to treat this loss as negligible; i.e.,
% \[
% \lim_{\delta_{l_1,l_2} \gg \tfrac{1}{B}} \bm{\mathcal{I}_{\tau_2,\tau_L}^{[l]}} = 0.
% \]

\begin{figure}[!ht]
    \centering
   \includegraphics[width=0.8\linewidth]{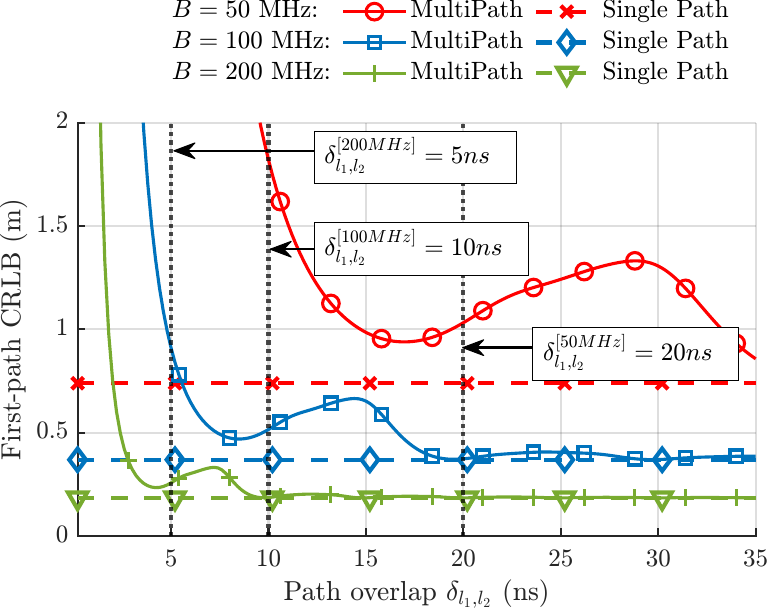}
  \caption{First-path CRLB for TOF-based ranging in a two-path multipath channel as a function of delay difference $\delta_{l_1,l_2}$. As the overlap increases, the multipath CRLB approaches the single-path bound. Larger bandwidths achieve this convergence at smaller overlaps; for $B=200$~MHz, the bound is nearly single-path limited at $\delta_{l_1,l_2}\approx 5$~ns, corresponding to about $1.5$~m of path-length separation in free space.}
  \label{fig_crlb_tof_estimation}
\end{figure}
Lemma 1 gives the conditional delay FIM given the path gains. In a full-parameter model, the gains are nuisance parameters; however, for resolvable wideband paths, the ranging information reduces to that of the isolated first path and depends only on its effective bandwidth and SNR \cite{shenetal}. Thus, \eqref{eq_ranging_uncertainty_diffraction_path} remains valid for the isolated diffraction path even with unknown complex gain. For nonresolvable MPCs, $\lambda_{k,n}$ should instead be computed from the full-parameter EFIM \cite{shenetal}.

\begin{remark}
\label{remark_resolvability_snr}
In the resolvable-path regime, i.e., $\delta_{l_1,l_2}\gg 1/B$, the ranging information for the isolated diffraction path depends only on its bandwidth and SNR, as captured by \eqref{eq_ranging_uncertainty_diffraction_path}. When paths overlap significantly, an overlap-aware EFIM is needed before generating the anchor-target dictionary in Remark~\ref{remark_anchor_node_dictionary}.
\end{remark}

% \begin{remark}
% \label{remark_resolvability_snr}
% Thus, for a fixed bandwidth that is sufficiently large relative to the inter-path delay separations, i.e., when $\delta_{l_1,l_2} \gg 1/B$, the MPCs become effectively resolvable. In this regime, the accuracy of TOF-based ranging estimation depends solely on the SNR of the diffraction path \eqref{eq_SNR_diffraction_path} and is unaffected by the presence of the other MPCs. Conversely, when the MPCs are not fully resolvable due to path overlap, an appropriate overlap-aware ranging-information model is needed to account for the resulting information loss when generating the anchor-target dictionary in Remark~\ref{remark_anchor_node_dictionary}.
% \end{remark}

% \begin{remark}
% \label{remark_resolvability_snr}
% Thus, for a sufficiently large fixed bandwidth, the accuracy of TOF-based ranging estimation depends solely on the SNR of the diffraction path \eqref{eq_SNR_diffraction_path} and is unaffected by the presence of the other MPCs. Conversely, in case there is path overlap, we need an appropriate path loss model that includes the effect of the path overlap on the SNR to generate the anchor-target dictionary in Remark~\ref{remark_anchor_node_dictionary}. 
% \end{remark}    

\vspace{-1em}

\section{Unified FIM for Mixed LOS/NLOS O2I Scenario}
In this section, we assume that the target’s $z$-coordinate, or equivalently its floor level, is known a priori (e.g., from barometric sensing or FirstNet \cite{Firstnetroadmap}), and focus on 2D in-floor localization using $K$ TOA-based range measurements. We first derive the FIM of the ranging measurements under the assumption that the diffraction path has been isolated from the received signal. We then use reparameterization to obtain the 2D localization FIM and the corresponding CRLB. Finally, we introduce an $(S_n,r_n)$ reformulation of this FIM, which reveals its geometric structure and underpins the single-target and multi-target anchor placement frameworks.

\subsection{FIM for LOS/NLOS Range Measurements for $K$ Anchors}

Let the $K$ ranging measurements between a fixed target $n$ and the $K$ anchors be collected in the vector $\bm{r}_n \in \mathbb{R}^K$, modeled as
\begin{equation}
    \bm{r}_n = \bm{p}_n + \bm{n}_n,
    \label{eq_range_measurement_model}
\end{equation}
where
\[
\bm{p}_n = [p_{1,n}, \dots, p_{K,n}]^T \in \mathbb{R}^K
\]
is the noiseless range vector, with $p_{k,n}$ denoting the diffraction-path length from the $k^{\mathrm{th}}$ anchor to target $n$, as given by \eqref{eq_diffraction_path_length_approx}. The vector $\bm{n}_n \in \mathbb{R}^K$ models the uncertainty in the TOF-based ranging measurements.

Under the assumed resolvability of the MPCs on each anchor--target link and the use of orthogonal waveforms across the $K$ anchors, we assume that an efficient TOF estimator is used for each link, so that the ranging error attains the corresponding CRLB. Under this assumption, the extracted ranging errors are modeled as Gaussian, and the covariance matrix of $\bm{n}_n$ is given by the diagonal matrix
\begin{align}
\bm{\Sigma}_n
= \mathbb{E}[\bm{n}_n \bm{n}_n^T]
= \operatorname{diag}\left(\tfrac{1}{\lambda_{1,n}}, \dots, \tfrac{1}{\lambda_{K,n}}\right),
\end{align}
where $\mathbb{E}[\cdot]$ denotes statistical expectation, and $\lambda_{k,n}$, defined in \eqref{eq_ranging_uncertainty_diffraction_path}, is the reciprocal of the CRLB on the variance of the $k^{\text{th}}$ ranging measurement. Thus, each diagonal entry of $\bm{\Sigma}_n$ represents the minimum achievable ranging-error variance for the corresponding anchor--target link under the adopted signal model.

Since the Gaussian likelihood associated with \eqref{eq_range_measurement_model} is parameterized by the mean vector $\bm{p}_n$, the FIM corresponding to the range parameter vector $\bm{p}_n$ is the inverse of the covariance matrix, given by
\begin{align}
    \bm{\mathcal{I}}_{\bm{p}_n}
    = \bm{\Sigma}_n^{-1}
    = \operatorname{diag}\left(\lambda_{1,n}, \dots, \lambda_{K,n}\right)
    \in \mathbb{R}^{K \times K}.
    \label{eq_FIM_ranging_measurements}
\end{align}
This matrix quantifies the information carried by the $K$ LOS/NLOS range measurements before reparameterizing the model in terms of the target location.

\subsection{2D Localization FIM in LOS/NLOS}
\label{section_2DFIM_formulation}

Although the targets lie in 3D space across different building floors, we assume that the target floor is known (e.g., from a barometer). Hence, for each target $n$, only the 2D in-floor position
\[
\bm{N}_n = [x_n,y_n]^T \in \mathbb{R}^2
\]
is estimated. Let
\[
\bm{p}_n(\bm{N}_n) = [p_{1,n},\ldots,p_{K,n}]^T \in \mathbb{R}^K
\]
denote the noiseless LOS/NLOS range vector. Since the Gaussian likelihood in \eqref{eq_range_measurement_model} is parameterized by $\bm{p}_n$, we apply Definition~\ref{def_FIM_transformation} with $\bm{\eta}=\bm{p}_n$ and $\bm{\xi}=\bm{N}_n$. The FIM for the $n^{\mathrm{th}}$ target location is therefore
\begin{align}
\label{eq_FIM_position_transformed_from_FIM_ranging_measurements}
\bm{\mathcal{I}}_{n}^{\langle 2\mathrm{D}\rangle}
=
\bm{J}_n^T \bm{\mathcal{I}}_{\bm{p}_n} \bm{J}_n,
\end{align}
where $\bm{\mathcal{I}}_{\bm{p}_n}$ is given by \eqref{eq_FIM_ranging_measurements}, and
\begin{align}
\bm{J}_n
=
\frac{\partial \bm{p}_n}{\partial \bm{N}_n}
=
\begin{bmatrix}
\bm{g}_{1,n}^T \\
\vdots \\
\bm{g}_{K,n}^T
\end{bmatrix}
\in \mathbb{R}^{K\times 2}
\end{align}
is the Jacobian of the noiseless range vector with respect to the 2D target position. Here,
\begin{align}
\bm{g}_{k,n}
=
\begin{bmatrix}
\frac{\partial p_{k,n}}{\partial x_n} \\
\frac{\partial p_{k,n}}{\partial y_n}
\end{bmatrix}
=
\left[
\frac{x_n-x_k}{p_{k,n}},
\frac{\sqrt{y_k^2+(z_k-z_n)^2}+y_n}{p_{k,n}}
\right]^T.
\label{eq_2D_information_vectors}
\end{align}

Substituting \eqref{eq_FIM_ranging_measurements} into \eqref{eq_FIM_position_transformed_from_FIM_ranging_measurements}, the 2D localization FIM can be written as
\begin{align}
\label{eq_2D_FIM_rank_1_outer_product}
\bm{\mathcal{I}}_{n}^{\langle 2\mathrm{D}\rangle}
=
\sum_{k=1}^{K}\lambda_{k,n}\,\bm{g}_{k,n}\bm{g}_{k,n}^T.
\end{align}
Thus, each anchor contributes a rank-one information matrix weighted by its ranging-information weight $\lambda_{k,n}$.

To estimate $\bm{N}_n$ uniquely, we require
\[
\rank(\bm{J}_n)=\dim(\bm{N}_n)=2.
\]
Equivalently, the information vectors $\{\bm{g}_{k,n}\}_{k=1}^K$ must span $\mathbb{R}^2$.

\begin{remark}
For the FIM to be full rank, at least $K=2$ anchors are required, and their information vectors $\bm{g}_{k,n}$ must be linearly independent. Equivalently, two suitably placed anchors are sufficient to estimate the unknown 2D position of a target on a known building floor.
\end{remark}

\subsection{$S_n,r_n$ FIM Reformulation}
\label{subsection_Sn_rn_reformulation}
We begin with the observation that \eqref{eq_2D_information_vectors} expresses the 2D FIM as a sum of scaled rank-one outer products of unit-norm vectors. Hence, each information vector can be parameterized by an angle $\psi_{k,n}$ as
\begin{equation}
\bm{g}_{k,n}
=
\begin{bmatrix}
\cos\psi_{k,n}\\
\sin\psi_{k,n}
\end{bmatrix}.
\label{eq_unit_vector_angle_parameterization}
\end{equation}

This parameterization is not unique to our path model but also applies to other 2D localization modalities; see Appendix~\ref{appendix_generality_Sn_rn} for more details.

For the proposed path model, the angle $\psi_{k,n}$ is given by
\begin{equation}
\label{eq_psi_definition}
\psi_{k,n}
=
\arctan2\!\left(
\sqrt{y_k^2+(z_k-z_n)^2}+y_n,\,
x_n-x_k
\right).
\end{equation}

Let $\lambda_{k,n}$ denote the information weight defined in \eqref{eq_ranging_uncertainty_diffraction_path}. For later use, define
\begin{equation}
\begin{aligned}
\label{eq_definition_perimeter}
S_n &\triangleq \sum_{k=1}^K \lambda_{k,n}, \\
u_n &\triangleq \sum_{k=1}^K \lambda_{k,n}\cos(2\psi_{k,n}), \\
v_n &\triangleq \sum_{k=1}^K \lambda_{k,n}\sin(2\psi_{k,n}),\\
r_n &\triangleq \sqrt{u_n^2+v_n^2}
= \left|\sum_{k=1}^K \lambda_{k,n}e^{j2\psi_{k,n}}\right|. 
\end{aligned}
\end{equation}

Here, $S_n$ represents the aggregate information strength, while $r_n$ captures the geometry induced by the anchor configuration. In our ranging model, $S_n$ is positively correlated with the received SNR through the weights $\lambda_{k,n}$.

Substituting \eqref{eq_unit_vector_angle_parameterization} into \eqref{eq_2D_FIM_rank_1_outer_product} and using the double-angle identities gives
\begin{equation}
\label{eq_2D_FIM_perimeter_form}
\bm{\mathcal{I}}_{n}^{\langle 2\mathrm{D}\rangle}
=
\frac{1}{2}
\begin{bmatrix}
S_n+u_n & v_n \\
v_n & S_n-u_n
\end{bmatrix}.
\end{equation}
As will become clear in the subsequent sections, this representation naturally leads to the polygon-closure interpretation through the quantities $S_n$ and $r_n$.

Assuming $\bm{\mathcal{I}}_{n}^{\langle 2\mathrm{D}\rangle}\succeq \mathbf{0}$, its eigenvalues are
\begin{equation}
\label{eq_2D_FIM_perimeter_form_eigenvalues}
\mu_{\pm}=\frac{S_n\pm r_n}{2},
\end{equation}
and therefore the eigenvalues of the corresponding CRLB matrix are
\begin{equation}
\label{eq_2D_CRLB_perimeter_form_eigenvalues}
\mu_{\pm}^{-1}=\frac{2}{S_n\pm r_n}.
\end{equation}
This form will be used extensively in the subsequent optimality analysis.

\section{Anchor Optimality for a Single Target}
In this section, we study anchor optimality for a single target using the \((S_n,r_n)\)-based FIM reformulation developed in the previous section. Beyond its independent theoretical interest, the single-target setting serves as a simple example problem that helps build intuition for the more general multi-target framework developed later. To expose the underlying geometric structure, we temporarily assume distance-independent path loss, so that the ranging-information weights \(\lambda_{k,n}\) are fixed and do not depend on the anchor locations, as in \cite{sadeghi2020optimal}. Under this assumption, the aggregate information term \(S_n\) is fixed, and the A-, D-, and E-optimality criteria reduce to functions of the residual term \(r_n\). This leads directly to a polygon-closure interpretation in the complex plane and clarifies when the three criteria coincide. In the subsequent sections, we relax this assumption and return to the general multi-target setting with distance-dependent path loss.

\subsection{A-, D- and E-Optimality Objectives}
\label{section_optimality_objectives}
% --- Confidence ellipsoid from the FIM ---
For a FIM \(\mathcal{I}\in\mathbb{R}^{d\times d}\),
the CRLB on the error covariance is \(\mathcal{I}^{-1}\).
The quadratic form \(\boldsymbol{e}^\top \mathcal{I}\,\boldsymbol{e}\) is \(\chi^2_d\)-distributed \cite{kay1993fundamentals},
yielding the \((1-\alpha)\) confidence ellipsoid
\begin{equation}
\label{eq_error_ellipsoid}
\mathcal{E}
=\left\{\boldsymbol{e}\in\mathbb{R}^d:\;
\boldsymbol{e}^\top \mathcal{I}\,\boldsymbol{e}
\le \chi^2_{d,1-\alpha}\right\}.
\end{equation}
Let \(\{\mu_i\}_{i=1}^d\) be the eigenvalues of \(\mathcal{I}^{-1}\).
The ellipsoid’s principal axes align with the eigenvectors of \(\mathcal{I}^{-1}\),
and the semi-axis lengths are
\[
a_i=\sqrt{\mu_i\,\chi^2_{d,1-\alpha}},\qquad i=1,\dots,d.
\]
Here, \(\chi^2_{d,1-\alpha}\) is a scalar that scales the CRLB ellipsoid to contain the estimation error
with probability \(1-\alpha\) according to the confidence level $\alpha$. In 2D, \(\chi^2_{2,0.95}=5.991\) gives the 95\% confidence ellipse.

Optimal design criteria minimize scalar performance measures derived from the CRLB matrix \(\mathcal{I}^{-1}\), which represents the fundamental error covariance in localization, by appropriately selecting anchor positions.

\paragraph{\emph{A-optimality}} 
The A-optimal design criterion seeks to minimize the trace of the CRLB (i.e., the sum of the variances). This corresponds to minimizing the sum of the eigenvalues in \eqref{eq_2D_CRLB_perimeter_form_eigenvalues}, yielding
\begin{align}
\label{eq_A_optimal_objective}
\min \; \Phi_A =& \min \;\operatorname{tr}\!\left( \left(\bm{\mathcal{I}_{n}^{\langle2D \rangle}}\right)^{-1}\right) \\ \nonumber
=& \min \;(\mu_+^{-1} + \mu_-^{-1}) 
= \min \;\frac{4S_n}{S_n^2-r_n^2}.
\end{align}
This quantity represents the sum of the error variances along orthogonal directions, and its square root is known as the \emph{position error bound (PEB)}. The PEB provides a fundamental lower bound on the root-mean-square error (RMSE) of any unbiased localization estimator, thereby directly linking A-optimality to minimizing localization error in the RMSE sense.

\paragraph{\emph{D-optimality}} corresponds to minimizing the determinant of the CRLB matrix or minimizing the volume of the positioning error ellipsoid (2D ellipse in our case). Equivalently, this is the product of the eigenvalues of the CRLB matrix in \eqref{eq_2D_CRLB_perimeter_form_eigenvalues}:
\begin{align}
\label{eq_D_optimal_objective}
\min \; \Phi_D 
&= \min \;\det\!\left(\left(\bm{\mathcal{I}_{n}^{\langle2D \rangle}}\right)^{-1}\right) \\ \nonumber
&= \min \;(\mu_+^{-1}\mu_-^{-1}) 
= \min \;\frac{4}{S_n^2 - r_n^2}.
\end{align}

This criterion seeks to minimize the volume of the error ellipsoid in \eqref{eq_error_ellipsoid}. In localization terms, this criterion reduces the overall volume of the confidence region.

\paragraph{\emph{E-optimality}} seeks to minimize the maximum eigenvalue of the CRLB matrix in \eqref{eq_2D_CRLB_perimeter_form_eigenvalues} and yields  
\begin{align}
\label{eq_E_optimal_objective}
\min \Phi_E
&= \min \; \operatorname{\text{eig}}_{\max}\!\left(\left(\bm{\mathcal{I}_{n}^{\langle 2\mathrm{D}\rangle}}\right)^{-1}\right) \\ \nonumber
&= \min \; \max\{\mu_+^{-1},\, \mu_-^{-1}\} \\ \nonumber
&= \min \; \frac{2}{S_n - r_n}
\end{align}
  
which corresponds to reducing the longest axis of the error ellipsoid. This ensures that the localization error is not excessively large in any direction, thereby improving robustness by improving worst-case error orientations.  

\subsection{Optimization Formulation for a Single Target With Distance-Independent Path Loss}

\par
Let the position of the $k^{\text{th}}$ anchor be represented by $\bm{A}_k = [x_k, y_k, z_k]^T \in \mathbb{R}^3,\;k\in\mathcal{K}$. The set of optimal anchor positions for a fixed target $n$ is $\{\bm{A}_k^{\star}\}_{k\in\mathcal{K}}$ and can be obtained as 
\begin{align}
\label{eq_anchor_optimality_objective_general}
&\text{Optimal Solution:}&\{\bm{A}_k^{\star}\}_{k\in\mathcal{K}}  = \arg\min_{\{\bm{A}_k\}_{k\in\mathcal{K}}} \Phi_X\left(\{\bm{A}_k\}_{k\in\mathcal{K}}\right) \\  \nonumber
&\text{Optimal Objective:}&\Phi_X^{\star} = \Phi_X\left(\{\bm{A}_k^{\star}\}_{k\in\mathcal{K}}\right).  
\end{align}
Here, $\Phi_X$ denotes one of the A-, D-, or E-optimal objectives as defined in \eqref{eq_A_optimal_objective}, \eqref{eq_D_optimal_objective}, and \eqref{eq_E_optimal_objective}, respectively. The target location is assumed to be fixed and known, so the 2D FIM depends only on the $K$ anchor positions $\{\bm{A}_k\}_{k\in\mathcal{K}}$. the 2D FIM is given by
\begin{align}
\label{eq_2DFIM_anchor_pos}
\bm{\mathcal{I}_{n}^{\langle2D \rangle}\!\left(\{\bm{A}_k\}_{k\in\mathcal{K}}\right)} 
&= \sum_{k=1}^{K} \bm{\mathcal{I}_{k,n}\!\left(\bm{A}_k\right)}, \\ \nonumber
\bm{\mathcal{I}_{k,n}\!\left(\bm{A}_k\right)} 
&= \lambda_{k,n}\!\left(\bm{A}_{k}\right)\,
      \bm{g}_{k,n}\!\left(\bm{A}_{k}\right)\,
      \bm{g}_{k,n}^T\!\left(\bm{A}_{k}\right).
\end{align}
Here, $\bm{\mathcal{I}_{k,n}}\left(\bm{A}_k\right)$ is the Fisher information contribution towards target $n$ from the $k^{\text{th}}$ anchor and is determined by both the ranging information weight $\lambda_{k,n}$ in \eqref{eq_ranging_uncertainty_diffraction_path} and the column vector of the Jacobian of the diffraction path model $\bm{g}_{k,n}$ capturing the LOS/NLOS geometry.

Recall that under the assumptions of resolvability and fixed signal bandwidth as observed in remark \ref{remark_resolvability_snr}, $\lambda_{k,n}$ depends solely on the SNR between the anchor and the target, with the SNR being the only term influenced by the anchor position. We temporarily take $\lambda_{k,n}$ to be anchor-position invariant (distance-independent path loss) as in \cite{sadeghi2020optimal}. Hence, $\lambda_{k,n}\!\left(\bm{A}_k\right) = \lambda_{k,n}$. 

\begin{thm}
Under the distance-independent path loss assumption, the ranging information weights 
$\{\lambda_{k,n}\}_{k=1}^K$ are treated as constants. Consequently, the quantity 
$S_n$ in \eqref{eq_definition_perimeter}, signifying signal coverage, is fixed for the $K$ anchors irrespective of their locations. For the FIM defined in \eqref{eq_2DFIM_anchor_pos}, the corresponding lower 
bounds on the A-, D-, and E-optimality objectives for the $n^{\text{th}}$ target are
\begin{equation}    
\Phi_A^{\star} = \frac{4}{S_n}, 
\qquad 
\Phi_D^{\star} = \frac{4}{S_n^2}, 
\qquad 
\Phi_E^{\star} = \frac{2}{S_n},
\end{equation}
where $S_n = \sum_{k=1}^K \lambda_{k,n}$.
\end{thm}

\begin{proof}
Starting from the A-, D-, and E-optimality objectives in 
\eqref{eq_A_optimal_objective}, \eqref{eq_D_optimal_objective}, 
and \eqref{eq_E_optimal_objective}, note that $S_n$ is fixed and $S_n>0$. 
Each criterion $\Phi_X$ is strictly increasing in $r \in [0,S_n)$. 
Hence, the minimum of each objective is attained at the smallest 
admissible value of $r_n$, namely $r_n = 0$. Substituting $r_n=0$ 
into the respective expressions yields the desired result.
\end{proof}
\begin{remark}
Under the distance-independent path loss assumption we have all three optimal objectives attain their optimum at $r_n=0$. Consequently, any anchor geometry that optimizes one of the A-, D-, or E-optimality criteria by achieving $r_n=0$ will also optimize the other two. Note there may be multiple optimal solutions as explained in Section \ref{section_closed_polygon_condition_fixed_snr}.
\end{remark}

% \begin{proof}
% Note, the 2D information vector $\bm{g}_k$ defined in \eqref{eq_2D_information_vectors} is unit norm. Let the $x$ and $y$ components be referred to as $[\bm{g}_k]_x$ and $[\bm{g}_k]_y$. Now, we can rewrite $[\bm{g}_k]_x=\cos{\psi_k}$ and $[\bm{g}_k]_y=\sin{\psi_k}$, where $\psi_k$ is the angle if $\bm{g}_k$ was represented on the complex plane.
% Now, after some algebraic manipulations we have 
% \begin{align}
% \Tr\left( \bm{J}_{\mathrm{2D}}^{-1}   \right)  = 
%  \frac{\sum_{i=1}^{N}\lambda_i}{AB-C^2},
% \end{align}
% where,
% \[
% A = \sum_{k=1}^K\lambda_k\cos{\psi_k}^2, B = \sum_{k=1}^K\lambda_k\sin{\psi_k}^2, C = \sum_{k=1}^K\lambda_k \frac{\sin{2\psi_k}}{2}.
% \]
% To minimize this quantity we need to maximize the denominator $AB-C^2$ which forms our new objective. 
% Now our maximization objective can be rewritten as $D=AB-C^2$ over the angles $\psi_k$. From the identity $\cos{\psi_k}^2+\sin{\psi_k}^2=1$ we have
% \[
% A+B = \sum_{k=1}^{K}\lambda_k
% \]
% which is independent from the angles $\psi_i$. Note that to maximize $D$, we need to maximize $AB$ subject to $A+B=\sum_{i=1}^{N}\lambda_i$ whilst simultaneously minimizing $C^2$. This is achieved when
% \begin{align}
% \label{eq_optimality_conditions}
% A=B=\frac{1}{2}\sum_{i=1}^{N}\lambda_i, \quad C^2 = 0    
% \end{align}

% which yields $D_{\text{max}}=\frac{1}{4}\left(\sum_{i=1}^N\lambda_i\right)^2$ and this completes the proof.
% \end{proof}

\subsection{Conditions for Achieving Optimality: Polygon Closure}
\label{section_closed_polygon_condition_fixed_snr}

\begin{figure}[htbp]
    \centering
    % --- Subfigure 1 ---
    \begin{subfigure}[t]{0.24\textwidth}
        \centering
        \includegraphics[width=\linewidth]{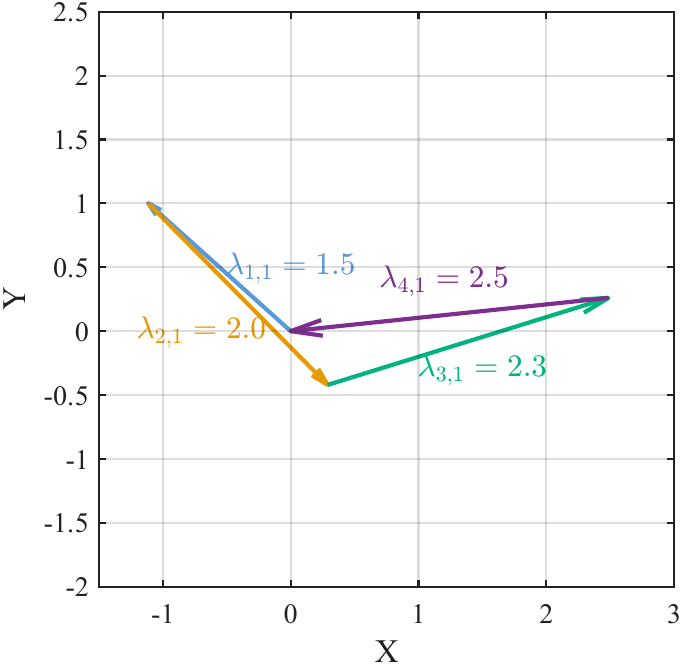}
        \caption{Closed Polygon construction}
        \label{fig_closed_polygon}
    \end{subfigure}
    \hfill
    % --- Subfigure 2 ---
    \begin{subfigure}[t]{0.24\textwidth}
        \centering
        \includegraphics[width=\linewidth]{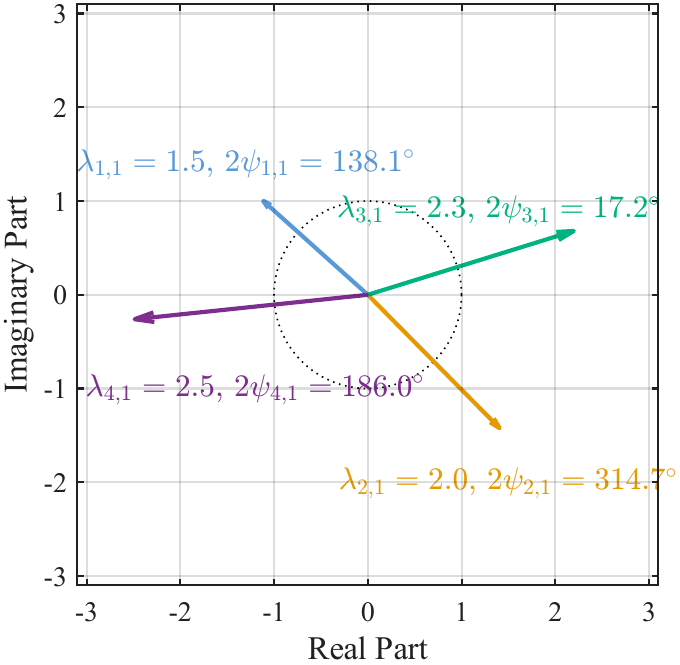}
        \caption{Complex Plane construction}
        \label{fig_complex_plane}
    \end{subfigure}
    \hfill

    % --- Main caption ---
\caption{For the fixed target \(n=1\) with ranging-information weights \(\{\lambda_{k,1}\}_{k\in\{1,2,3,4\}}=\{1.5,\,2,\,2.3,\,2.5\}\), the generalized triangle inequality in Lemma~\ref{lemma_necessary_sufficient_condition_polygon_construction} holds. Our construction therefore yields the polygon-closing angles \(2\psi_{1,1}=138.2^\circ\), \(2\psi_{2,1}=314.7^\circ\), \(2\psi_{3,1}=17.2^\circ\), and \(2\psi_{4,1}=186^\circ\), which simultaneously attain A-, D-, and E-optimality. The corresponding anchor locations are then recovered by solving the transcendental relation in \eqref{eq_psi_definition}, which arises from our mixed LOS/NLOS \((S_n,r_n)\)-based 2D FIM formulation. For the Euclidean ranging case, \cite{sadeghi2020optimal} derives a method which is effectively a specific polygon-closure procedure and is obtained by solving the corresponding transcendental equation from the \((S_n,r_n)\)-based 2D FIM for Euclidean ranging; see Appendix~\ref{appendix_generality_Sn_rn} for further details.}

    \label{fig_polygon_closure_problem}
\end{figure}

Observe that the condition $r_n=0$ directly implies $u_n=0$ and $v_n=0$ defined in \eqref{eq_definition_perimeter}. Equivalently, in terms of $\psi_{k,n}$ and $\lambda_{k,n}$ we get
\begin{align}
\label{eq_polygon_optimality_condition}
\Big\|\sum_{k=1}^K \lambda_{k,n} e^{j2\psi_{k,n}}\Big\|= 0. 
% \text{or}\\ \nonumber
% &\sum_{k=1}^K \lambda_{k,n} \sin(2\psi_{k,n}) = 0, \quad \sum_{k=1}^K \lambda_{k,n} \cos(2\psi_{k,n}) = 0.
\end{align}
Geometrically, the optimality condition can be interpreted in the complex plane as the construction of a closed $K$-sided polygon whose side lengths are given by $\{\lambda_{1,n},\dots,\lambda_{K,n}\}$. As an example we have explained this pictorially in \figref{fig_polygon_closure_problem}. This raises the natural question: for a fixed set of $\{\lambda_{k,n}\}$ corresponding to $K$ anchors, does there always exist a set of orientations $\{2\psi_{1,n},\dots,2\psi_{K,n}\}$ such that the condition in~\eqref{eq_polygon_optimality_condition} is satisfied? Equivalently, can one always construct a closed polygon in the complex plane with side lengths $\lambda_{k,n}$ and orientations $2\psi_{k,n}$? This problem is closely related to the generalized triangle inequality, which characterizes the necessary and sufficient conditions for the existence of such a polygon.

\begin{lemma}
\label{lemma_necessary_sufficient_condition_polygon_construction} Let $\{\lambda_{k,n}\}_{k \in \mathcal{K}}$ denote a set of nonnegative real numbers, derived from the ranging uncertainties in \eqref{eq_ranging_uncertainty_diffraction_path}, which represent the magnitudes of $K$ vectors in the complex plane $\mathbb{C}$. A necessary and sufficient condition for the existence of angles $\{\psi_k\}_{k\in\mathcal{K}}$ such that the vectors $\{\lambda_k e^{j2\psi_k}\}_{k\in\mathcal{K}}$ form a closed polygon is that the generalized polygon inequality holds:
\[
    \lambda_{k,n} \leq \sum_{i \in \mathcal{K}\setminus\{k\}} \lambda_{i,n},
    \quad \forall k \in \mathcal{K}.
\]
\end{lemma}

Since our FIM formulation accounts for mixed LOS/NLOS conditions, the angles \(\psi_{k,n}\) are not the physical bearing angles considered in \cite{sadeghi2020optimal}, but rather the angles associated with the information vectors in our 2D FIM formulation; see Section~\ref{subsection_Sn_rn_reformulation} and Appendix~\ref{appendix_generality_Sn_rn}. For Euclidean ranging, however, \(\psi_{k,n}\) coincides with the bearing angle, and \cite{sadeghi2020optimal} provides a method for solving the resulting transcendental equation to recover anchor locations. Likewise, in our mixed LOS/NLOS setting, once a set of angles \(2\psi_{k,n}\) satisfying \eqref{eq_polygon_optimality_condition} is obtained, the corresponding anchor locations may be recovered in principle from \eqref{eq_psi_definition}. Because this equation is transcendental, multiple solutions may exist, so the optimal anchor locations are generally not unique.

 \vspace{-1em}

\section{Extension to the Multi-Target Setting}
Having established the single-target optimality conditions and the associated polygon-closure interpretation, we now turn to the multi-target setting. Here, we reintroduce distance-dependent path loss, so that the ranging-information weights \(\lambda_{k,n}\) depend on the anchor locations rather than remaining fixed constants. This makes the optimization problem in \eqref{eq_anchor_optimality_objective_general} inherently nonconvex in the continuous anchor-position variables \(\{\bm{A}_k\}_{k\in\mathcal{K}}\). To address this difficulty, we discretize the feasible anchor region and reformulate the problem as a combinatorial anchor-selection problem. This provides a tractable framework for extending single-target anchor optimality to region-wide design over multiple targets.
 
\vspace{-1em}

\subsection{Anchor Discretization and Combinatorial Reformulation}
Instead of trying to optimally place $K$ anchors outside the building, we sample the 3D space outside the building into $M$ points represented by the index set $\mathcal{M}\triangleq\{1,\cdots,M\}$. Each point $m\in\mathcal{M}$ represents a candidate anchor location and our optimization problem can be reformulated as selecting $K$ anchors from the candidate anchor set $\mathcal{M}$ such that our optimality criteria are met.
The anchor selection can be modeled using binary variables $x_m \in \{0,1\}$, where $x_m=1$ indicates that candidate anchor $m \in \mathcal{M}$ is chosen. 
For each target $n\in\mathcal{N}$ and candidate anchor location indexed by $m \in \mathcal{M}$, we can reformulate the coverage and the geometry term defined in \eqref{eq_2D_FIM_perimeter_form} for the discretized scenario as follows
\begin{align}
\label{eq_definition_perimeter_gap_definition}
S_n(\{x_m\}_{m\in\mathcal{M}}) &\triangleq \sum_{m \in \mathcal{M}} \lambda_{m,n} x_m, \\ \nonumber
r_n(\{x_m\}_{m\in\mathcal{M}}) &\triangleq \sqrt{u_n^2+v_n^2}
= \Bigg|\sum_{m \in \mathcal{M}} \lambda_{m,n} e^{j2\psi_{m,n}} x_m\Bigg|.
\end{align}
Observe in the above reformulation, $\lambda_{m,n}$ and $\psi_{m,n}$ are now fixed scalars that represent the ranging information weight between the $n^{\text{th}}$ target and $m^{\text{th}}$ candidate anchor location and our optimization variables are the binary variables $\{x_m\}_{m \in {M}}$. Moreover, $S_n$ is affine in the binary variables, while $r_n^2$ is quadratic in them.
With the coverage term and geometry terms defined in \eqref{eq_definition_perimeter_gap_definition}, we can now derive the 2D FIM for the discretized setting using \eqref{eq_2D_FIM_perimeter_form} and the eigenvalues of the CRLB matrix in \eqref{eq_2D_CRLB_perimeter_form_eigenvalues} to yield the three optimality objectives for each target $n\in\mathcal{N}$ as
\begin{align}
\label{eq_A_optimal_node_n}
\Phi_{A,n}^{\star} &= \min_{\{x_m\}_{m\in\mathcal{M}}}\; \frac{4S_n}{S_n^2-r_n^2} \equiv \max_{\{x_m\}_{m\in\mathcal{M}}}\; \frac{S_n^2-r_n^2}{4S_n}, \\
\label{eq_D_optimal_node_n}\
\Phi_{D,n}^{\star} &= \min_{\{x_m\}_{m\in\mathcal{M}}}\; \frac{4}{S_n^2-r_n^2}\equiv\max_{\{x_m\}_{m\in\mathcal{M}}}\; \frac{S_n^2-r_n^2}{4}, \\
\label{eq_E_optimal_node_n}
\Phi_{E,n}^{\star} &= \min_{\{x_m\}_{m\in\mathcal{M}}}\; \frac{2}{S_n-r_n}\equiv\max_{\{x_m\}_{m\in\mathcal{M}}}\; \frac{S_n-r_n}{2}
\end{align}
Note that under the continuous relaxation $x_m \in [0,1]$, the E-opt objective $S_n - r_n$ is concave, and the D-opt formulation is implemented by maximizing a quantity proportional to $\sqrt{\det(\mathbf{I}_n)}=\tfrac{1}{2}\sqrt{S_n^2-r_n^2}$, which is concave over $\mathbf{I}_n \succeq \mathbf{0}$. This structure is well suited to commercial MISOCP solvers using branch-and-bound; see Section
~\ref{section_branch_and_bound} for details.

\begin{remark}
\label{remark_anchor_node_dictionary}
We have the target locations fixed and candidate anchors sampled with a spatial resolution \(d\) from the feasible region (outside the building). We precompute \(\lambda_{m,n}\) and \(\psi_{m,n}\) for all pairs \((m,n)\) using \eqref{eq_ranging_uncertainty_diffraction_path} and \eqref{eq_psi_definition}. These fixed values form an \textbf{anchor–target dictionary}. This is fed as an input to the optimizer (solver) which then selects anchors \(\{x_m\}_{m\in\mathcal{M}}\) that optimize the objective within the feasible region. In essence, the dictionary captures geometry and signal coverage information, while the optimizer selects the best anchors.
\end{remark}

Discretizing the feasible region thus converts the continuous placement task into a combinatorial \emph{anchor selection problem}, in which the objective is to select $K$ anchors from $M$ candidates to optimize the chosen criterion across all $N$ fixed indoor targets.

\subsection{Multi-Target Objectives: Global and Min–Max Formulations}
\label{section_global_vs_minimax_objective_multi_node}
After the discretization step there are two ways to extend our single target objectives to multi-target. Recall, for a network of \(N\) targets, each target \(n\) is associated with a local 2D position estimation FIM denoted by 
\[\bm{\mathcal{I}_{n}^{\langle 2D \rangle}}\!\big(\{x_m\}_{m\in\mathcal{M}}\big)\in\mathbb{S}_{++}^{2\times2},\]
determined by the anchor configuration \(\{x_m\}_{m\in\mathcal{M}}\).  
The overall FIM for estimating the 2D position for all targets is block-diagonal and is given by
{\small
\begin{align}
    \bm{\mathcal{I}_{\mathrm{glob}}}\!\big(\{x_m\}_{m\in\mathcal{M}}\big)
    \;=\;
    \operatorname{blkdiag}\!\Big(
        \bm{\mathcal{I}_{1}^{\langle 2D \rangle}}\!,
        \dots,
        \bm{\mathcal{I}_{N}^{\langle 2D \rangle}}\!
    \Big),
\end{align}
}
whose inverse and determinant factor across targets as 
\[
\bm{\mathcal{I}_{\mathrm{glob}}^{-1}}
=\operatorname{blkdiag}\!\big((\bm{\mathcal{I}_{1}^{\langle 2D \rangle})^{-1}},\dots,(\bm{\mathcal{I}_{N}^{\langle 2D \rangle})^{-1}}\big)
\]
and  
\[\!\det \bm{\mathcal{I}_{\mathrm{glob}}}
=\prod_{n=1}^{N}\!\det \bm{\mathcal{I}_{n}^{\langle 2D \rangle}}.
\]
Accordingly, the classical \emph{global} optimal design criteria aggregate information across all targets:
{\small
\begin{align}
    \text{A-opt:}\quad 
        &\min_{\{x_m\}_{m\in\mathcal{M}}}\;
        \Phi_A^{\mathrm{glob}}
        =\operatorname{tr}\!\big(\bm{\mathcal{I}_{\mathrm{glob}}^{-1}}\big)
        =\sum_{n=1}^{N}\operatorname{tr}\!\big((\bm{\mathcal{I}_{n}^{\langle 2D \rangle})^{-1}}\big),
        \\[0.4em]
    \text{D-opt:}\quad 
        &\min_{\{x_m\}_{m\in\mathcal{M}}}\;
        \Phi_D^{\mathrm{glob}}
        =\det (\bm{\mathcal{I}_{\mathrm{glob}}^{-1}}) \\ \nonumber 
        &\qquad\qquad\qquad\;\;\;=\prod_{n=1}^{N}\!\det ((\bm{\mathcal{I}_{n}^{\langle 2D \rangle})^{-1}}),
        \\[0.4em]
        \label{eq_E-opt_global}
    \text{E-opt:}\quad 
        &\min_{\{x_m\}_{m\in\mathcal{M}}}\;
        \Phi_E^{\mathrm{glob}}
          =\operatorname{\text{eig}}_{\max}\!\big((\bm{\mathcal{I}_{\mathrm{glob}})^{-1}}\big) \\ \nonumber
        &\qquad\qquad\qquad\;\;\;=\max_{n}\operatorname{\text{eig}}_{\max}\!\big((\bm{\mathcal{I}_{n}^{\langle 2D \rangle})^{-1}}\big).
\end{align}
}
While the global A- and D-objectives minimize the \emph{average} or \emph{aggregate} estimation error across all targets, they do not explicitly control the weakest-performing target.  
To ensure fairness and robustness, we use a \emph{min–max} formulation instead. This optimizes the worst-case per-target objective and is given by
{\small
\begin{align}
    \text{A-opt (min–max):}\quad 
        &\min_{\{x_m\}_{m\in\mathcal{M}}}\;
        \max_{n}\operatorname{tr}\!\big((\bm{\mathcal{I}_{n}^{\langle 2D \rangle})^{-1}}\big),
        \\[0.3em]
    \text{D-opt (min–max):}\quad 
        &\min_{\{x_m\}_{m\in\mathcal{M}}}\;
    \max_{n}\!\det\!\big((\bm{\mathcal{I}_{n}^{\langle 2D \rangle})^{-1}}\big),
        \\[0.3em]
        \label{eq_min_max_E_opt}
    \text{E-opt (min–max):}\quad 
        &\min_{\{x_m\}_{m\in\mathcal{M}}}\;
        \max_{n}\operatorname{\text{eig}}_{\max}\!\big((\bm{\mathcal{I}_{n}^{\langle 2D \rangle})^{-1}}\big),
\end{align}
}
where the last expression \eqref{eq_min_max_E_opt} coincides exactly with the global E-optimal objective \eqref{eq_E-opt_global} because the smallest eigenvalue of the block-diagonal FIM equals the minimum of the smallest eigenvalues of its constituent blocks.  
Thus, E-optimality naturally enforces fairness across targets, whereas A- and D-optimality require explicit min–max reformulations to prevent sacrificing accuracy at poorly conditioned targets. Hence, for the rest of this paper, we adopt a min–max formulation that minimizes the maximum per-target optimality objective.

\vspace{-1em}

\subsection{Relationship Among A-, D-, and E-Optimal Objectives}

Recall that the per-target D-opt objective in \eqref{eq_D_optimal_node_n} and E-opt objective in \eqref{eq_E_optimal_node_n} were concave in the continuous relaxation of the optimization variables. In this section, we show that the E-optimal objective bounds the A-optimal objective, while the D-optimal objective is related to the
A-optimal objective through the aggregate information term \(S_n\). The following proposition formalizes these relationships among the A-, D- and E-optimal objectives.

\begin{proposition}[Relationship between A-, D- and E-optimality objectives]
\label{proposition_A_D_E_single_node}
For each target $n\in\mathcal{N}$, we have
\begin{enumerate}
\item[(i)] (\emph{E straddles A}) For all $n$,
\begin{equation}\label{eq:E-straddles-A}
    \Phi_{E,n} \;\le\; \Phi_{A,n} \;\le\; 2\,\Phi_{E,n}.
\end{equation}
\item[(ii)] (\emph{D--A scaling}) For all $n$,
\begin{equation}\label{eq:A-equals-S-times-D}
    \Phi_{A,n} \;=\; S_n\,\Phi_{D,n}.
\end{equation}
Thus, if $S_n$ is fixed across feasible designs, any design that
decreases $\Phi_{D,n}$ also decreases $\Phi_{A,n}$.
Moreover, if \(S_{\min}\le S_n \le S_{\max}\) over all feasible anchor selections,
where \(S_{\min}\) and \(S_{\max}\) bound the aggregate information
strength \(S_n\), then 
\begin{equation}\label{eq:S-bounded}
    S_{\min}\,\Phi_{D,n} \;\le\; \Phi_{A,n} \;\le\; S_{\max}\,\Phi_{D,n}.
\end{equation}
Hence, bounded $S_n$ yields only a multiplicative comparison bound, not
a monotonic guarantee.
\end{enumerate}
\end{proposition}

\begin{proof}
For (i), the A-opt per-target objective can be rewritten as
\[
\Phi_{A,n}
=\frac{4S_n}{S_n^2-r_n^2}
=\frac{2}{S_n-r_n}+\frac{2}{S_n+r_n}
=\Phi_{E,n}+\frac{2}{S_n+r_n},
\]
and since $0\le r_n\le S_n$ implies $S_n-r_n\le S_n+r_n$, we get
\[
\frac{2}{S_n+r_n}\le \frac{2}{S_n-r_n}=\Phi_{E,n}.
\]
Therefore,
\[
\Phi_{E,n}\;\le\;\Phi_{A,n}\;\le\;\Phi_{E,n}+\Phi_{E,n}
=2\,\Phi_{E,n},
\]
which proves \eqref{eq:E-straddles-A}.

For (ii), this follows directly from the definitions of A- and D-optimality in \eqref{eq_A_optimal_node_n} and \eqref{eq_D_optimal_node_n}.
\end{proof}

\begin{remark}[Extension to the multi-target case]
% Proposition~\ref{proposition_A_D_E_single_node} extends directly to the multi-target setting using the min--max formulation, by taking the maximum over all targets $n \in \mathcal{N}$. 
The D--A relation in Proposition~\ref{proposition_A_D_E_single_node} remains per-target; in the min--max multi-target setting, bounds on
\(S_n\) provide only multiplicative comparison bounds, not monotonic
guarantees.
\end{remark}
Now we bring all of these ideas together in two algorithms with exact D- and E-optimal objectives.

\subsubsection{\textbf{Algorithm 1: E-opt MISOCP}}
We seek to select exactly $K$ anchors such that the E-optimal objective in \eqref{eq_E_optimal_node_n} is maximized for the worst-performing target. In other words the E-optimal objective aims to maximize the difference between the polygon perimeter $S_n$ and the residual $r_n$. The selected anchors describe, for each target $n$, a $K$-sided polygon in the information domain. The resulting optimization problem can be formulated as

\begin{ieeeproblem}[breakable=false, width=1\columnwidth, fontupper=\small, left=4pt, right=4pt, top=4pt, bottom=4pt]{E-opt-MISOCP}
\begin{align}
&\max_{\{x_m\}_{m\in\mathcal{M}}} \;\;  \gamma \label{eq_soc_objective} \\
\text{s.t.}\;\;
& \sum_{m=1}^M x_m = K,\quad x_m\in\{0,1\}, \label{eq_card_gap_soc}\\
& S_n = \sum_{m=1}^M \lambda_{m,n} x_m,\quad 
  p_n = \sum_{m=1}^M \lambda_{m,n}\cos\!\big(2\psi_{m,n}\big)x_m, \nonumber\\
& q_n = \sum_{m=1}^M \lambda_{m,n}\sin\!\big(2\psi_{m,n}\big)x_m,\quad \forall n\in \mathcal{N}, \label{eq_S_pq_def_soc}\\
& r_n^2 \ge p_n^2 + q_n^2,\quad \forall n\in \mathcal{N}, \label{eq_soc_norm}\\
& \gamma \;\le\; S_n - r_n,\quad \forall n\in \mathcal{N}, \label{eq_gamma_gap_soc}\\
& r_n \ge 0,\;\; \gamma \ge 0,\quad \forall n\in \mathcal{N}. \label{eq_bounds_soc}
\end{align}
\end{ieeeproblem}

The objective is to maximize a scalar $\gamma$ that lower bounds the E-optimal criterion for every target $n\in\mathcal{N}$, as enforced by constraint \eqref{eq_gamma_gap_soc}. The optimization variables are the binary indicators $x_m\in\{0,1\}$, $m\in\mathcal{M}$, where $x_m=1$ indicates that candidate anchor $m$ is selected. Constraint \eqref{eq_card_gap_soc} imposes the cardinality requirement that exactly $K$ anchors are selected from the $M$ candidate locations. For each target $n\in\mathcal{N}$, the aggregate information strength $S_n$ is computed using \eqref{eq_S_pq_def_soc}, while the residual geometry term $r_n$ is defined through the second-order cone constraint \eqref{eq_soc_norm}. Constraint \eqref{eq_gamma_gap_soc} then ensures that $\gamma$ lower bounds the E-optimal objective in \eqref{eq_E_optimal_node_n} for all targets. Finally, constraint \eqref{eq_bounds_soc} enforces the nonnegativity of $\gamma$ and $r_n$.
 
Notice that the optimization problem features a linear objective and linear constraints, with the exception of \eqref{eq_soc_norm}, which imposes a convex quadratic constraint. This structure classifies the problem as a second-order cone program (SOCP). Furthermore, since the anchor selection variables $x_m$ are restricted to be binary, the resulting formulation is a mixed-integer second-order cone program (MISOCP).
 
\subsubsection{\textbf{Algorithm 2: D-opt MISOCP}}
From the D-optimality criterion in \eqref{eq_D_optimal_node_n}, maximizing $\tau$ is equivalent to maximizing the worst-case determinant of the FIM, across all targets, i.e.
\[
\det(\bm{\mathcal{I}_n}) \;\;\ge\;\; \frac{\tau^2}{4},\quad \forall n \in \mathcal{N}.
\]
Thus the D-optimal algorithm can be formulated as

\begin{ieeeproblem}[breakable=false, width=1\columnwidth, fontupper=\small, left=4pt, right=4pt, top=4pt, bottom=4pt]{D-opt-MISOCP}
\begin{align}
&\max_{\{x_m\}_{m\in\mathcal{M}}} \;\;  \tau \nonumber\\
\text{s.t.}\;\;
& \sum_{m=1}^M x_m = K,\quad x_m\in\{0,1\}, \label{eq_card_det_soc}\\
& S_n = \sum_{m=1}^M \lambda_{m,n} x_m,\quad 
  p_n = \sum_{m=1}^M \lambda_{m,n}\cos\!\big(2\psi_{m,n}\big)x_m, \nonumber\\
& q_n = \sum_{m=1}^M \lambda_{m,n}\sin\!\big(2\psi_{m,n}\big)x_m,\quad 
  \forall n\in \mathcal{N}, \label{eq_S_pq_def_det}\\
& r_n^2 \ge p_n^2 + q_n^2,\quad 
  \forall n\in \mathcal{N}, \label{eq_soc_norm_det}\\
& a_n = S_n - r_n,\quad 
  b_n = S_n + r_n,\quad 
  a_n,b_n \ge 0,\;\forall n\in\mathcal{N}, \label{eq_ab_def}\\
& \sigma_n^2 \le 2 a_n b_n,\quad 
  \sigma_n \ge 0,\;\forall n\in\mathcal{N}, \label{eq_rotated_soc}\\
& \sqrt{2}\;\tau \le \sigma_n,\quad 
  \forall n\in \mathcal{N}, \label{eq_tau_link}\\
& \tau \ge 0, r_n \ge 0. \label{eq_bounds_det}
\end{align}
\end{ieeeproblem}

The objective is to maximize a common scalar number $\tau$ over the indicator variables $x_m \in \{0,1\},\; m \in \mathcal{M}$. 
Constraint \eqref{eq_card_det_soc} enforces the cardinality condition, ensuring that exactly $K$ anchors are selected out of the available $\mathcal{M}$. 
For each target $n\in \mathcal{N}$, the perimeter $S_n$ and residual side length $r_n$ are defined in \eqref{eq_S_pq_def_det}--\eqref{eq_soc_norm_det}. 
The auxiliary variables in \eqref{eq_ab_def} satisfy
\[
a_n b_n = (S_n-r_n)(S_n+r_n)=S_n^2-r_n^2.
\]
Using the standard rotated SOC form, constraint \eqref{eq_rotated_soc} implies
\[
\sigma_n^2 \le 2a_n b_n = 2\bigl(S_n^2-r_n^2\bigr),
\]
or equivalently,
\[
\frac{\sigma_n}{\sqrt{2}} \le \sqrt{S_n^2-r_n^2}.
\]
Constraint \eqref{eq_tau_link} then imposes
\[
\sqrt{2}\,\tau \le \sigma_n,\qquad \forall n\in\mathcal{N},
\]
so that, together with \eqref{eq_rotated_soc},
\[
\tau \le \sqrt{S_n^2-r_n^2},\qquad \forall n\in\mathcal{N}.
\]
Since
\[
\det(\bm{\mathcal{I}_n})=\frac{1}{4}(S_n^2-r_n^2),
\]
it follows that
\[
\det(\bm{\mathcal{I}_n}) \ge \frac{\tau^2}{4},\qquad \forall n\in\mathcal{N}.
\]
Hence, maximizing $\tau$ is equivalent to maximizing the worst-case determinant across all targets.

The problem therefore has a linear objective and linear constraints, except for \eqref{eq_soc_norm_det} and \eqref{eq_rotated_soc}, which are convex quadratic constraints. This places the problem in the class of SOCP. Furthermore, because the anchor selection variables $x_m$ are binary, the final formulation is a MISOCP.

\subsection{Branch-and-Bound for MISOCP}
\label{section_branch_and_bound}
Mixed-integer formulations such as MISOCP are typically solved via branch-and-bound~\cite{land1960automatic,conforti2014integer}. The method branches on binary selection variables ($\{x_m\}_{m \in \mathcal{M}}$ in our case) to build a search tree; at each node, the binary constraints are relaxed to continuous values which yields a convex subproblem— an SOCP for MISOCP—whose optimal value provides a valid upper bound. We refer to the \emph{incumbent} as the best feasible solution found so far. These bounds drive pruning: any node whose relaxation cannot outperform the incumbent is discarded, so only promising regions of the search space are explored. 

\textbf{Optimality guarantees.} For the current discretized and sampled anchor-target dictionary, branch-and-bound certifies epsilon optimality once the gap between the incumbent and the best relaxation bound is at or below a user specified tolerance, assuming exact node relaxations and correct pruning or exhaustive enumeration~\cite{benson2016mixed}. In this setting, the certificates are: (i) \textbf{Feasibility} the incumbent satisfies all constraints including integrality; (ii) \textbf{epsilon optimality certificate} the incumbent objective matches or is within $\epsilon$ of the global bound from the relaxations, that is $\mathrm{gap} \triangleq (\mathrm{bound} - \mathrm{incumbent}) \le \epsilon$ with the usual sign convention for minimization or maximization; and (iii) \textbf{Infeasibility certificate} the search can prove that no feasible solution exists within the discretized sample region.

\subsection{Interpretation as a Polygon Closure Problem}
In the multi-target case, achieving optimality is no longer equivalent to closing the polygon for each target individually. 
For instance, under the E-optimal objective, maximizing the denominator \(S_n - r_n\) requires simultaneously maximizing 
the coverage term \(S_n\) and minimizing the residual \(r_n\). 
However, since all targets share the same set of anchors, it may not be possible to close the polygon for every target. 
Moreover, the polygon closure itself may not be achievable, as the generalized triangle inequality in Lemma \ref{lemma_necessary_sufficient_condition_polygon_construction} may not hold for 
anchors that yield the highest signal coverage. Observe in \figref{fig_polygon_closure_multi_target}, the optimizer balances the tradeoff between maximizing overall coverage $S_n$ and achieving favorable 
geometric configurations $r_n$ across all targets. 
Geometrically, for each target \(n \in \mathcal{N}\), \(S_n\) represents the \emph{perimeter} of a polygon with sides 
\(\{\lambda_{m,n} : x_m = 1\}\) oriented by \(\{2\psi_{m,n} : x_m = 1\}\) in the complex plane, 
while \(r_n\) denotes the \emph{residual}, i.e., the length of the polygon’s unclosed side.

\begin{figure}[!htbp]
    \centering
   \includegraphics[width=0.5\linewidth, trim=0mm 0mm 0mm 0mm, clip]{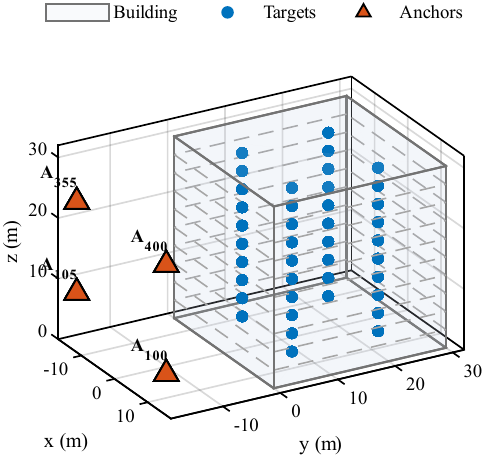}
  \caption{Optimal \(K=4\) anchor configuration returned by our optimization framework, namely E-opt-MISOCP, for the mixed LOS/NLOS scenario, using a candidate anchor dictionary sampled with \(d=4\) m. Under LOS/Euclidean path models, the usual intuition is to place anchors roughly around the building perimeter to obtain uniform angular coverage. In contrast, the resulting optimum here exhibits a rotated rhombus-like arrangement near one side of the building, reflecting the different placement geometry induced by the mixed LOS/NLOS setting. The numbers next to the anchors specify the anchor index selected from the anchor-target dictionary as the solution}
  \label{fig_spatial_plot}
\end{figure}

\begin{figure}[!htbp]
    \centering
   \includegraphics[width=\linewidth, trim=8mm 90mm 8mm 90mm, clip]{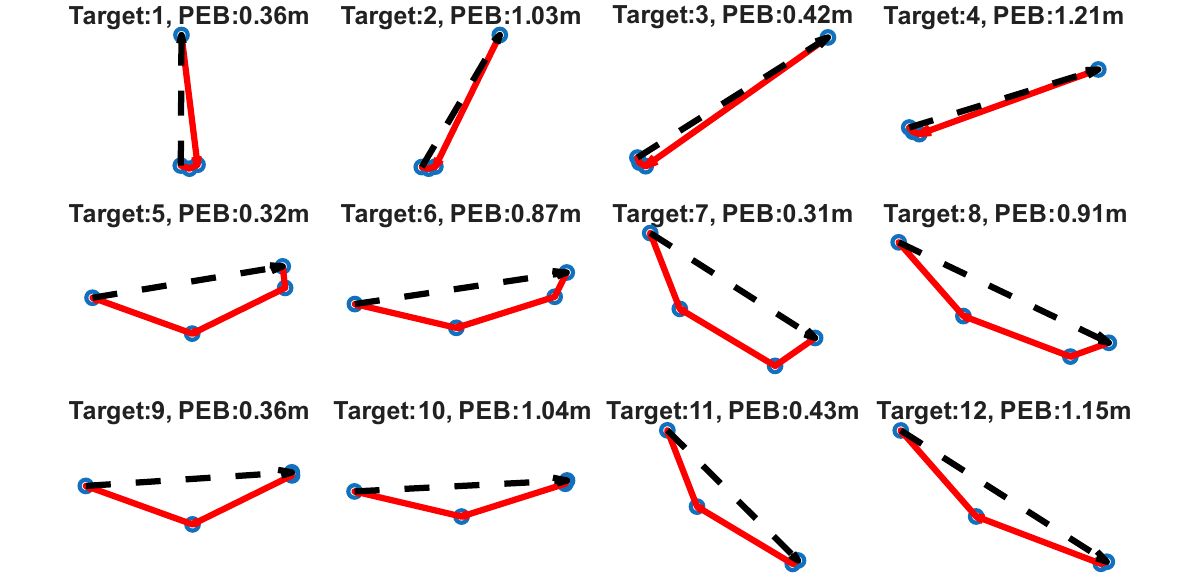}
  \caption{Polygon closure for \(12\) representative targets located on \(3\) of the building's \(10\) floors. The common \(4\)-anchor configuration is obtained by applying the E-opt-MISOCP jointly to all \(40\) targets; only three floors are shown for visual clarity. For each target, the formulation promotes near-closure of the corresponding polygon by reducing the residual \(r_n\), while simultaneously favoring a large aggregate information strength \(S_n\). Equivalently, it maximizes the worst-target value of \(S_n-r_n\), thereby improving worst-case localization performance and, consequently, the A-optimality metric (or PEB).}
\label{fig_polygon_closure_multi_target}
\end{figure}

\section{Results and Discussion}
\subsection{Performance Metrics}
\label{section_metrics}
Let \(\{x_m\}_{m\in\mathcal{M}}^\star\) denote the solution to the optimization problem. Define the worst-performing target 
\(n_{\mathrm{worst}}\in\mathcal{N}\). 
For this target, compute \(S_n\) and \(r_n\) via \eqref{eq_definition_perimeter_gap_definition}, and the corresponding CRLB matrix as
\[\bm{\Sigma}_{n_{\mathrm{worst}}}^{\star}\coloneqq \left(\bm{\mathcal{I}_{n_{worst}}^{\langle2D \rangle}}\right)^{-1}\]

from \eqref{eq_2D_FIM_perimeter_form}.
We report three scalar measures of localization accuracy:
{\small
\begin{align}
\label{eq_PEB_metrics}
\mathrm{PEB}_{n_{worst}}^{\star}      &= \sqrt{\operatorname{tr}\!\big(\bm{\Sigma}_{n_{\mathrm{worst}}}^{\star}\big)}, \\
\label{eq_CER_metrics}
\rho_{n_{worst}}^{\star}              &= \sqrt{5.991\,\sqrt{\det\!\big(\bm{\Sigma}_{n_{\mathrm{worst}}}^{\star}\big)}}, \\
\label{eq_MAD_metrics}
\sigma_{\max,n_{worst}}^{\star}     &= \sqrt{\lambda_{\max}\!\big(\bm{\Sigma}_{n_{\mathrm{worst}}}^{\star}\big)}.
\end{align}
}
\noindent
Here, $\mathrm{PEB}_{n^{\star}}$ (PEB) gives the root mean squared (RMS) position error from the CRLB via $\sqrt{\operatorname{tr}(\bm{\Sigma}_n^\star)}$;
$\rho_{n^{\star}}$ (confidence ellipsoid radius (CER)) is the radius of a $95\%$ equal-area circle for the error ellipse, using $\chi^2_{2,0.95}=5.991$; 
and $\sigma_{\max,n^{\star}}$ (maximum axis deviation (MAD)) is the worst-direction standard deviation error obtained from the largest eigenvalue of $\bm{\Sigma}_n^\star$.
All three have units of length and are scaled versions of the A-,D- and E-optimality criteria respectively.

\subsection{Simulation Methodology}
We assess performance via a Monte~Carlo study with \(100\) trials. Target locations are fixed. In each trial \(t\), we first form a uniform 3D grid of candidate anchors \(\{\bm{A}_m\}_{m\in\mathcal{M}}\) with spacing \(d\) meters along each of the \(x,y,z\) axes. We then apply a Cranley--Patterson random shift \cite{cranley1976randomization} by drawing a vector on every trial $t$ such that
\(\bm{u}^{(t)}=[u_x,u_y,u_z]^\mathsf{T}\) with \(u_i \overset{\text{i.i.d.}}{\sim} \mathcal{U}([-\tfrac{d}{2},\tfrac{d}{2}])\), shifting every grid point by \(\bm{u}^{(t)}\), and wrapping any out-of-domain points to within the feasible region. Using this candidate set, we precompute the anchor--target dictionary parameters \(\{\lambda_{m,n},\psi_{m,n}\}\), solve for the anchor selections \(\{x_m^{(t)}\}_{m\in\mathcal{M}}\), and compute the accuracy metrics. We report boxplots of each metric \eqref{eq_PEB_metrics},\eqref{eq_CER_metrics},\eqref{eq_MAD_metrics} for the worst-performing target across trials. The other simulation parameters are given in Table \ref{tab_system_params}.

\begin{table}[t]
\centering
\caption{System and Geometric Simulation Parameters}
\scriptsize
\begin{tabular}{lcl}
\toprule
\textbf{Parameter} & \textbf{Symbol} & \textbf{Value} \\
\midrule
\multicolumn{3}{l}{\textit{System Parameters}} \\[2pt]
Bandwidth & $B$ & $200~\text{MHz}$ \\
Transmit Power & $P_t$ & $30~\text{dBm}$ \\
Transmit Antenna Gain & $G_t$ & $10~\text{dBi}$ \\
Receive Antenna Gain & $G_r$ & $10~\text{dBi}$ \\
Carrier Frequency & $f_c$ & $10~\text{GHz}$ \\
Noise Figure & $F$ & $3~\text{dB}$ \\[4pt]
\multicolumn{3}{l}{\textit{Geometric Parameters}} \\[2pt]
Building parameter &L& $10$m\\
Number of Floors & $N_{\text{floor}}$ & $10$ \\
Floor Height & $h_{\text{floor}}$ & $3~\text{m}$ \\
Building Length & $L_x$ & $2L~\text{m}$ \\
Building Depth & $L_y$ & $3L~\text{m}$ \\
Target Positions & $\{\bm{N}_n\}_{n\in \mathcal{N}}$ & 4 fixed targets per floor \\
Anchor Feasible Region & $\{\bm{A}_m\}$ & 
$\substack{
[\text{-}L,\text{-3}L,\!\,0]^T
\le \bm{A}_m \le [L,0,3L]^T
}$ \\
Anchor sampling length &d& $\{5,4,3\}$m\\
Number of anchors &K&\{3,4,5\} \\
\bottomrule
\end{tabular}
\label{tab_system_params}
\end{table}

\subsection{Numerical Comparison of the Algorithms }
As a baseline, we use a greedy minimax heuristic with one-swap
local refinement over the precomputed anchor--target dictionary.
For each criterion, the algorithm first selects the best anchor
pair from the dictionary by exhaustive search. It then adds the
remaining anchors sequentially, each time selecting the candidate
that yields the smallest worst-case objective value across all
targets. After \(K\) anchors have been selected, a one-swap local
search is applied. In each iteration, one selected anchor is
replaced by an unselected candidate whenever the replacement
decreases the worst-case objective. This process is repeated until
no improving one-for-one replacement exists. The resulting three
baseline variants are denoted by GreedyMinimax-A-opt,
GreedyMinimax-D-opt, and GreedyMinimax-E-opt, corresponding to
the worst-case A-, D-, and E-optimal objectives, respectively.

We compare all algorithms using the three performance metrics defined in
\eqref{eq_PEB_metrics}--\eqref{eq_MAD_metrics}, namely the PEB, CER,
and MAD. We interpret these metrics from an optimal experimental design
perspective: the design variables are the anchor locations, and the goal
is to improve the FIM/CRLB induced by the selected anchors. Since PEB,
CER, and MAD are scaled versions of the A-, D-, and E-optimal
objectives, respectively, each algorithm is expected to perform best in
the metric corresponding to its design objective. Although these are
bound-based design metrics rather than outputs of a specific estimator,
companion estimator-level work for the same diffraction-based mixed
LOS/NLOS model demonstrates near-CRLB RMSE
performance~\cite{duggal2026newlocationestimatormixed}, supporting
their use as practical design targets.

Next, we evaluate the proposed methods in two ways. First, we vary the sampling distance $d\in\{5m,4m,3m\}$ of the anchor dictionary, which changes the size and resolution of the feasible candidate set and these results are shown in \figref{fig_positioning_results_vs_d}. Second, we vary the number of selected anchors, considering $K\in\{3,4,5\}$. The corresponding results for the $K$-sweep are shown in \figref{fig_positioning_results_vs_K}.

We first discuss the dependence on the sampling distance $d$. As shown in \figref{fig_CER_vs_d}, D-opt-MISOCP consistently achieves the lowest CER across all values of $d$, which is expected since CER is the metric corresponding to the D-optimal objective. Similarly, \figref{fig_MAD_vs_d} shows that E-opt-MISOCP consistently achieves the lowest MAD for every sampling distance, since MAD is the metric associated with the E-optimal objective. In contrast, PEB is not the exact objective optimized by either proposed MISOCP. Nevertheless, Proposition~\ref{proposition_A_D_E_single_node} shows that the E-optimal objective controls the A-optimal objective within a factor of two. 
Hence, the E-optimal criterion provides a certified surrogate for
improving A-opt performance, whereas the D-optimal criterion is only
multiplicatively related to PEB through \(S_n\). Improvements in
non-matched metrics are therefore empirical, not monotonic guarantees.
Next, we examine the effect of the number of selected anchors \(K\).
Increasing \(K\) generally improves performance, as reflected by lower
PEB, CER, and MAD for both algorithms. The gains beyond \(K=4\) are
smaller, suggesting that \(K=4\) provides a favorable
accuracy--complexity tradeoff in this setting. As in the $d$-sweep, each formulation performs best in the metric matched to its own design objective, with D-opt-MISOCP yielding the strongest CER performance in \figref{fig_CER_vs_k} and E-opt-MISOCP yielding the strongest MAD performance in \figref{fig_MAD_vs_k}, while both continue to provide competitive improvements in PEB as well in \figref{fig_PEB_vs_k}.

\subsection{Complexity Trends With Respect to $K$ and $d$}

After discretization, both E-opt-MISOCP and D-opt-MISOCP reduce to cardinality-constrained binary selection problems over the candidate-anchor set $\mathcal{M}$, in which exactly $K$ anchors are chosen from $M$ candidates. Both formulations are solved by branch-and-bound over the same binary decision variables $x_m$, $m \in \mathcal{M}$. A crude worst-case view therefore suggests exponential dependence on the number of candidate anchors $M$. However, in the present setting, a more informative measure of the search-space size is the number of feasible selections, namely $\binom{M}{K}$. Hence, for fixed sampling distance $d$ (and therefore fixed $M$), the computational burden grows combinatorially with $K$ and increases rapidly as more anchors are selected.

Next, for fixed $K$, the candidate anchors are generated by sampling a three-dimensional feasible region with spacing $d$. Ignoring boundary effects, the number of candidate anchors therefore scales as $M=\Theta(d^{-3})$. Hence, decreasing $d$ refines the feasible anchor dictionary and allows the solver to explore better configurations, but it also increases the number of binary decision variables and enlarges the underlying search space. As a result, the computational burden grows rapidly as $d$ decreases. Equivalently, since the dominant combinatorial term is $\binom{M}{K}$ and $M=\Theta(d^{-3})$, the complexity inherits this dependence through the growth of $M$.

Both formulations therefore share the same leading dependence on $K$ and $d$. The D-opt-MISOCP is typically somewhat heavier per node than the E-opt-MISOCP because it introduces additional auxiliary variables and an extra rotated second-order cone constraint, but this does not alter the dominant trend. In practice, branch-and-bound solvers use relaxations, pruning, incumbent heuristics, branching rules, and presolve, so these expressions should be interpreted only as rough worst-case indicators rather than exact runtime laws.
This trade-off is also reflected in the numerical results in \figref{fig_computation_d} and \figref{fig_computation_k}. Under the imposed $70$ s per-trial time limit, E-opt-MISOCP and D-opt-MISOCP are able to certify optimality for some trials when the dictionary is coarser (e.g., $d=5\,\mathrm{m}$). At finer sampling resolutions, the solvers may no longer certify optimality within the time limit; nevertheless, the incumbent solutions returned in practice remain high quality and consistently yield improved localization performance across all three metrics. Thus, although finer discretization increases computational burden, it still provides a practically useful accuracy-complexity trade-off. This trade-off between performance and computational tractability highlights the practical balance achieved by the proposed MISOCP formulations.

\begin{figure*}[htbp]
    \centering
    % --- Subfigure 1 ---
    \begin{subfigure}[t]{0.31\textwidth}
        \centering
        \includegraphics[width=\linewidth]{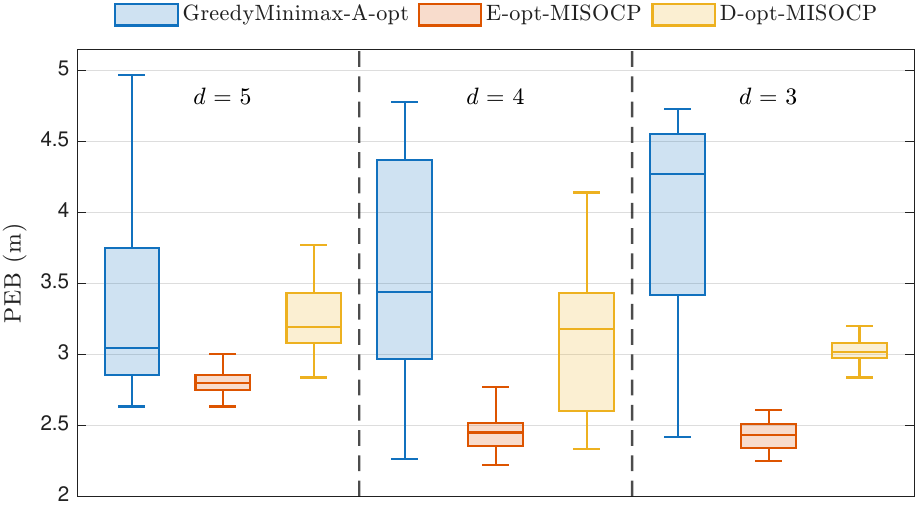}
        \caption{PEB plot}
        \label{fig_PEB_vs_d}
    \end{subfigure}
    \hfill
    % --- Subfigure 2 ---
    \begin{subfigure}[t]{0.31\textwidth}
        \centering
        \includegraphics[width=\linewidth]{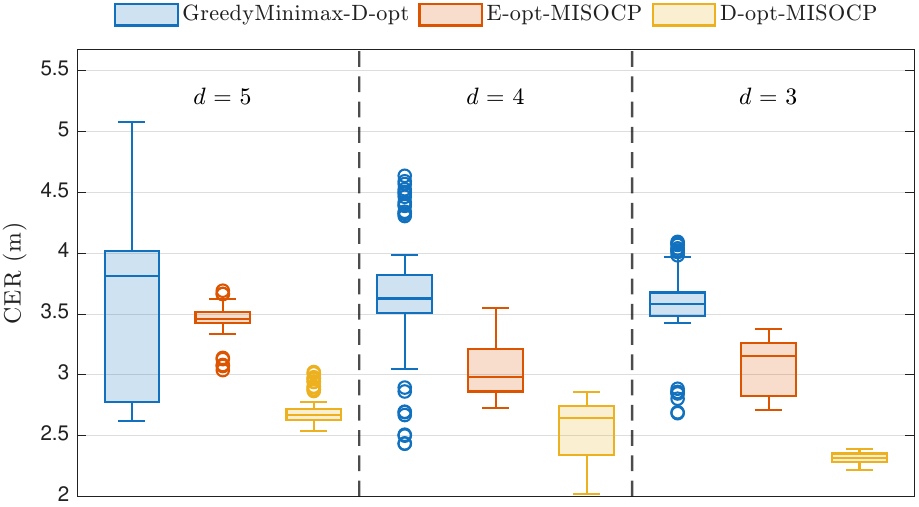}
        \caption{CER plot}
        \label{fig_CER_vs_d}
    \end{subfigure}
    \hfill
    % --- Subfigure 3 ---
    \begin{subfigure}[t]{0.31\textwidth}
        \centering
        \includegraphics[width=\linewidth]{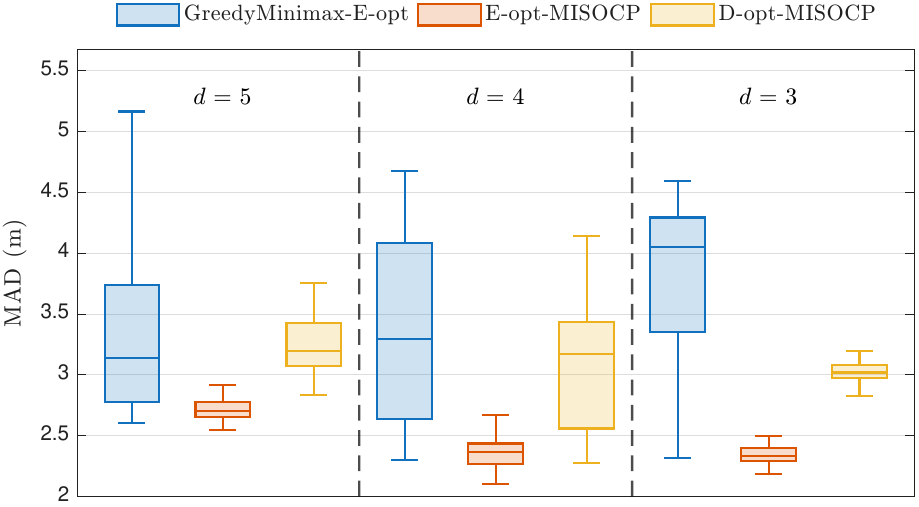}
        \caption{MAD plot}
        \label{fig_MAD_vs_d}
    \end{subfigure}

    % --- Main caption ---
   \caption{Localization performance of the E-opt-MISOCP, and D-opt-MISOCP formulations with \(K=4\) anchors versus various sampling distances \(d=3m ,4m ,5m \). Each boxplot summarizes variability across Monte Carlo trials: the horizontal line marks the median; the box spans the 25th--75th interquartile (IQR); whiskers extend to \(1.5\times\) IQR; and dots indicate outliers beyond the whiskers. Hence, shorter boxes and whiskers indicate tighter dispersion (greater robustness), while a lower median indicates better typical performance. For larger \(d\), the solver attains optimality within the 70\,s time limit; for denser dictionaries (\(d<5\) m), solutions are high quality but not provably optimal within the time limit. Among the objectives, D-opt-MISOCP achieves the lowest CER, E-opt-MISOCP attains the smallest MAD, and both also reduce PEB, consistent with the theoretical relationships among A-, D-, and E-opt criteria in Proposition~\ref{proposition_A_D_E_single_node}. The baseline performance is obtained using the GreedyMinimax heuristic with one-swap local refinement for each of the A-, D-,
and E-optimality criteria.}

    \label{fig_positioning_results_vs_d}
\end{figure*}

\begin{figure*}[htbp]
    \centering
    % --- Subfigure 1 ---
    \begin{subfigure}[t]{0.31\textwidth}
        \centering
        \includegraphics[width=\linewidth]{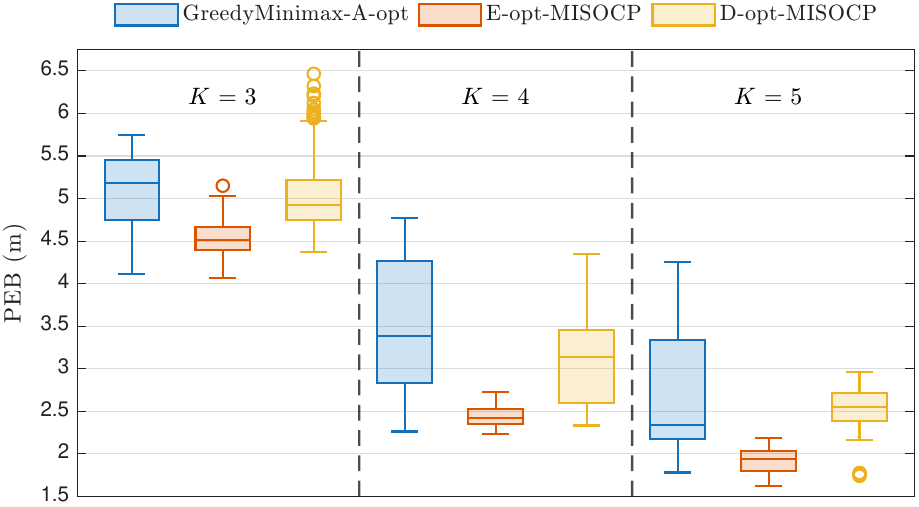}
        \caption{PEB plot}
        \label{fig_PEB_vs_k}
    \end{subfigure}
    \hfill
    % --- Subfigure 2 ---
    \begin{subfigure}[t]{0.31\textwidth}
        \centering
        \includegraphics[width=\linewidth]{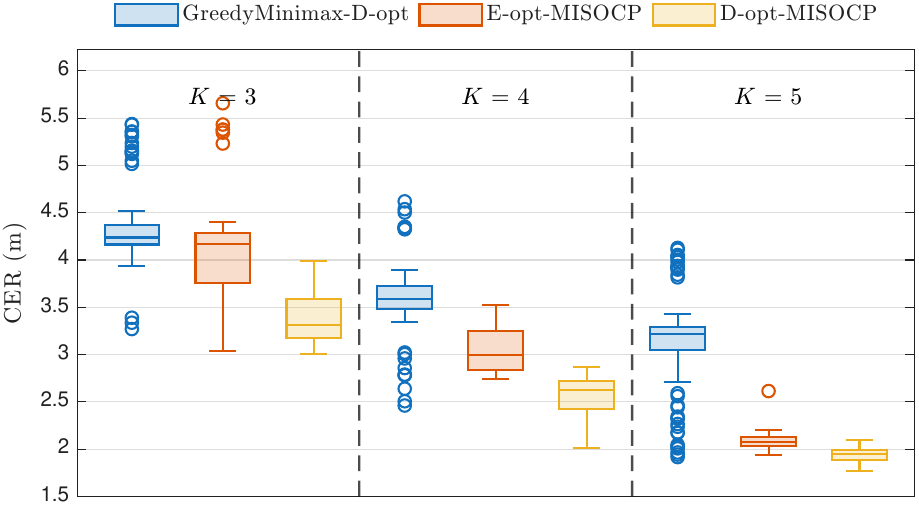}
        \caption{CER plot}
        \label{fig_CER_vs_k}
    \end{subfigure}
    \hfill
    % --- Subfigure 3 ---
    \begin{subfigure}[t]{0.31\textwidth}
        \centering
        \includegraphics[width=\linewidth]{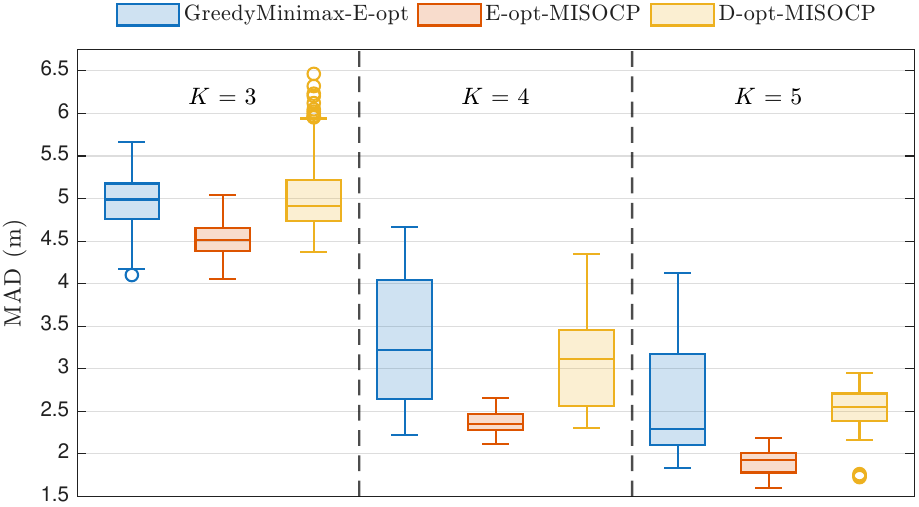}
        \caption{MAD plot}
        \label{fig_MAD_vs_k}
    \end{subfigure}

    % --- Main caption ---
   \caption{Localization performance of the E-opt-MISOCP and D-opt-MISOCP formulations as the number of selected anchors increases (increasing K). Each boxplot summarizes variability across Monte Carlo trials: the horizontal line marks the median; the box spans the 25\%--75\% interquartile range (IQR); whiskers extend to \(1.5\times\) IQR; and dots indicate outliers beyond the whiskers. Thus, a lower median indicates better typical performance, while shorter boxes and whiskers indicate lower variability and greater robustness. For both E-opt-MISOCP and D-opt-MISOCP, increasing the number of anchors \(K\) improves localization performance overall, leading to lower PEB, CER, and MAD. As expected, the formulation matched to a given metric achieves the strongest performance in that metric, with D-opt-MISOCP yielding the lowest CER and E-opt-MISOCP yielding the lowest MAD, while both also improve PEB as \(K\) increases. In contrast, the greedy baselines exhibit less consistent behavior across \(K\), particularly for MAD, where both the trends and the variability are less regular.}

    \label{fig_positioning_results_vs_K}
\end{figure*}

\begin{figure*}[htbp]
    \centering
    \begin{subfigure}[t]{0.48\textwidth}
        \centering
        \includegraphics[width=\linewidth]{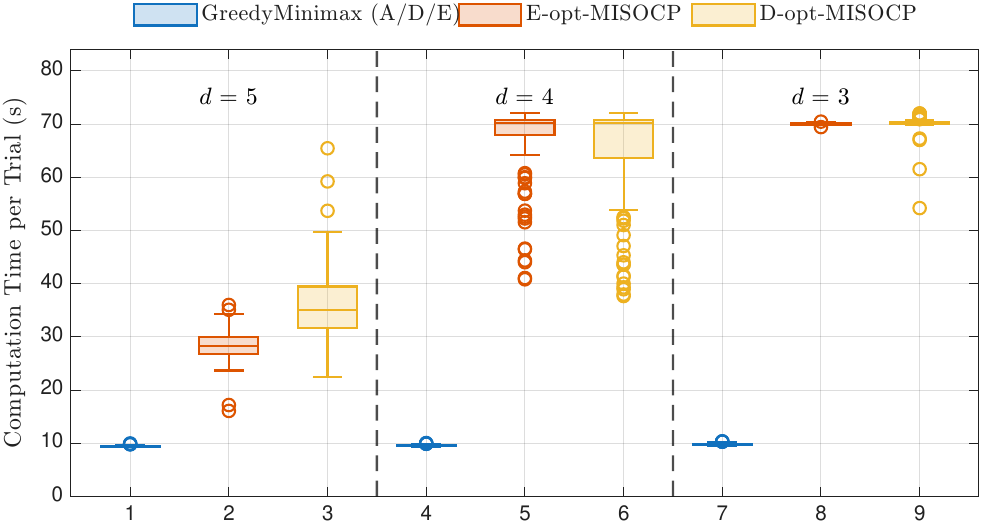}
        \caption{Solve time vs.\ sampling distance $d$}
        \label{fig_computation_d}
    \end{subfigure}
    \hfill
    \begin{subfigure}[t]{0.48\textwidth}
        \centering
        \includegraphics[width=\linewidth]{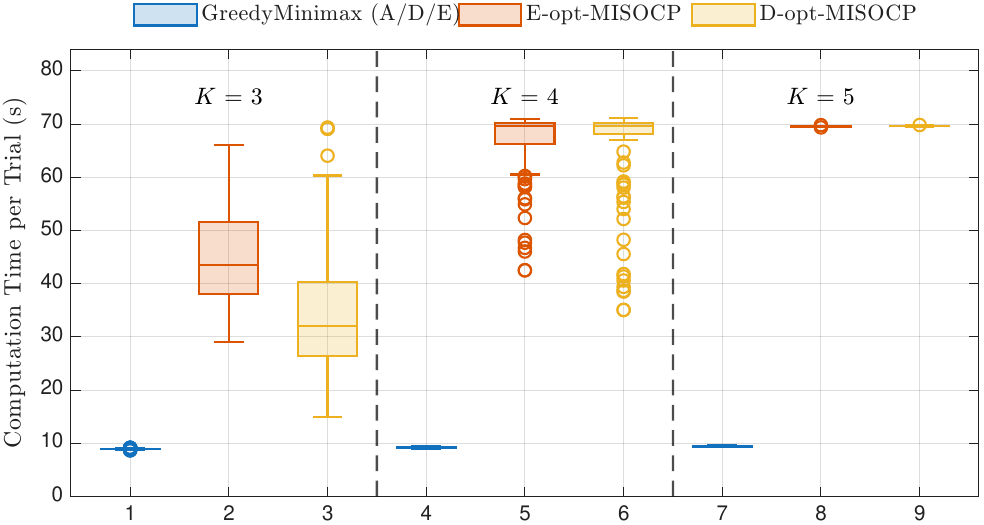}
        \caption{Solve time vs.\ number of anchors $K$}
        \label{fig_computation_k}
    \end{subfigure}
  \caption{Solve-time trends for E-opt-MISOCP and D-opt-MISOCP. (a) Finer sampling (smaller \(d\)) for the \textbf{anchor-target} dictionary defined in remark \ref{remark_anchor_node_dictionary} yields more candidate anchors and larger anchor-target dictionaries, leading to longer solve times. (b) Increasing the number of selected anchors \(K\) enlarges the cardinality-constrained search space, which likewise leads to longer solve times. In both cases, a \(70\,\text{s}\) per-trial time limit is imposed; if the solver does not terminate before this limit, it returns the best incumbent solution found up to that point.}
  \label{fig_computation_time}
\end{figure*}

\section*{Conclusion}
We presented a unified Fisher–information-based framework for optimal anchor placement in mixed LOS/NLOS O2I environments where NLOS propagation is diffraction-dominated. The framework couples a generalized diffraction path-length model that gracefully reduces to the Euclidean distance in LOS with a KED-based path-loss/SNR model, enabling a single treatment across propagation regimes. For the single-target case (with distance-independent path loss), we showed that A-, D-, and E-optimal designs coincide and admit a polygon-closure interpretation, yielding clear optimality conditions. We then extended the framework to the multi-target case via a discretization step and formulated min–max objectives, leading to tractable optimization programs with guarantees. The multi-target case can also be used to optimize over a feasible region. Numerical results illustrate the resulting anchor layouts, localization performance, and computational tradeoffs under the assumed regime of reliable diffraction-path isolation and resolvable diffraction paths in mixed LOS/NLOS propagation.

% \subsection{Varying the anchors}
% \begin{figure*}[htbp]
%     \centering
%     % --- Subfigure 1 ---
%     \begin{subfigure}[t]{0.31\textwidth}
%         \centering
%         \includegraphics[width=\linewidth]{figs/PEB_boxplot_vs_K.pdf}
%         \caption{PEB}
%         \label{fig_PEB}
%     \end{subfigure}
%     \hfill
%     % --- Subfigure 2 ---
%     \begin{subfigure}[t]{0.31\textwidth}
%         \centering
%         \includegraphics[width=\linewidth]{figs/CER_boxplot_vs_K.pdf}
%         \caption{CER}
%         \label{fig_CER}
%     \end{subfigure}
%     \hfill
%     % --- Subfigure 3 ---
%     \begin{subfigure}[t]{0.31\textwidth}
%         \centering
%         \includegraphics[width=\linewidth]{figs/MAD_boxplot_vs_K.pdf}
%         \caption{MAD}
%         \label{fig_MAD}
%     \end{subfigure}

%     % --- Main caption ---
%     \caption{}
%     \label{fig_positioning_results}
% \end{figure*}

% \begin{figure}[h]
%     \centering
%    \includegraphics[width=\linewidth]{figs/time_boxplot_vs_K.pdf}
%   \caption{}
%   \label{fig_computation_K}
% \end{figure}

% \subsection{Increasing W}
% \begin{figure}[h]
%     \centering
%    \includegraphics[width=\linewidth]{figs/MILP_vs_MISOCP_Eopt.pdf}
%   \caption{}
%   \label{fig_MILP_vs_MISOCP}
% \end{figure}
\vspace{-1em}
\appendix
\subsection{Diffraction Path Length Derivation}
\label{appendix_diffraction_path_length_derivation}
Observe from \figref{fig_system_model}, we have the $n^{\text{th}}$ target at $\bm{N}_n=[x_n,y_n,z_n]^T$, anchor indexed by $k$ at $\bm{A}_k=[x_k,y_k,z_k]^T$ and the diffracting edge $e \in \{u,l\}$. Define the length of the perpendiculars from the target to the diffraction edge $e$ as $r$ and from the anchor to the diffraction edge as $r_{\perp,e}$, we have
\[
r=\sqrt{y_n^2+\Delta^2},\quad r_{\perp,e}=\sqrt{y_k^2+(z_n-z_k+s_e\Delta)^2},
\]
where $e \in \{u,l\}, s_u=1,s_l=-1$.
The diffraction path length $p_{k,n}$ between the target and the $k^{\text{th}}$ anchor can be written as a sum of two distances $\bm{A}_k\bm{Q}+\bm{Q}\bm{N}_n$ and expanded to
\begin{align}
\label{eq_appendix_diffraction_path_length}
p_{k,n}^{\langle\text{e}\rangle} = \sqrt{(x_{k,n}^{\langle e \rangle}-x_n)^2+r^2} + \sqrt{(x_{k,n}^{\langle e \rangle}-x_k)^2+r_{\perp,e}^2}.
\end{align}

Using Fermat's principle of least time (or equivalently the diffraction law) \cite{duggal2025diffractionaidedwirelesspositioning} we set the first derivative of the path length with respect to the x-coordinate $x_{k,n}^{\langle e \rangle}$ of the diffraction point to zero i.e.
$\frac{\delta p_{k,n}^{\langle\text{e}\rangle}}{\delta x_{k,n}^{\langle e \rangle}} = 0$. We obtain the following expression
\begin{align}
\frac{(x_{k,n}^{\langle e \rangle}-x_n)^2}{(x_{k,n}^{\langle e \rangle}-x_n)^2+r^2} = \frac{(x_{k,n}^{\langle e \rangle}-x_k)^2}{(x_{k,n}^{\langle e \rangle}-x_k)^2+r_{\perp,e}^2}
\end{align}
After some algebraic manipulations we get two solutions
\[
x_{k,n}^{\langle e \rangle} = \frac{rx_k+r_{\perp,e}x_n}{r+r_{\perp,e}},\quad x_{k,n}^{\langle e \rangle} = \frac{rx_k-r_{\perp,e}x_n}{r-r_{\perp,e}}
\]
The first solution ensures $x_{k,n}^{\langle e \rangle}$ lies between $x_n$ and $x_k$ whereas the second solution lies outside this interval and can be ignored. This is a direct consequence of the Fermat's principle of least time since the first solution provides a smaller path length.
Next, we can substitute the first solution in \eqref{eq_appendix_diffraction_path_length} to obtain the final path length result in \eqref{eq_diffraction_path_length}.

\subsection{FIM of the Diffraction Path Delay}

\label{section_FIM_first_arriving_path}
Using the definition of the received signal in \eqref{eq_received_signal} and assuming Gaussian noise with power $N_0$, the log-likelihood function of the received signal \eqref{eq_received_signal} can be written as
\[
\log \chi(r(t) | \bm{\eta}) = -\frac{1}{N_0} \int_0^{T_{\text{ob}}} \left| r(t) - \sum_{l=1}^L h_l s(t - \tau_l) \right|^2 dt + \text{C}.
\]
Here, $C$ is a constant.
Let $\mathcal{L}\triangleq\{1,\dots,L\}$ denote the index set of the $L$ MPCs. Using the definition of the FIM in Definition~\ref{definition_FIM}, we obtain $(l_1,l_2)^{\text{th}}$ entry of the FIM as
{\small
\[
[\bm{\mathcal{I}_{\tau}}]_{l_1,l_2} = \frac{2}{N_0} h_{l_1}^{*}h_{l_2} \int_0^{T_{\text{ob}}} \left( \frac{d s(t - \tau_{l_1})}{dt} \right)^{*}\left( \frac{d s(t - \tau_{l_2})}{dt} \right) dt.
\]
}
Next, define $\delta_{l_1,l_2}=\tau_{l_1}-\tau_{l_2}$. We use Parseval's theorem and time-frequency properties of the Fourier transform to express the entries of the FIM as
{\small
\[
[\bm{\mathcal{I}_{\tau}}]_{l_1, l_2}
= \frac{2}{N_0} \operatorname{Re} \left\{ h_{l_1}^* h_{l_2} \int_{-\infty}^{\infty} (2\pi f)^2 |S(f)|^2 e^{j2\pi f\delta_{l_1,l_2}} df \right\}
\]
}
where \( S(f) \) is the Fourier transform of \( s(t) \).
To simplify the above expression, we use the assumption that the PSD is flat as in \eqref{eq_flat_PSD} and use the Fourier integral result
% \begin{align}
% \int_{-a}^{a} f^2 e^{j 2\pi f \Delta} df
% =& - \frac{1}{(2\pi)^2} \cdot \frac{d^2}{d\Delta^2} \left( \frac{\sin(2\pi a \Delta)}{\pi \Delta} \right) \nonumber \\ 
% =& \frac{2a \cos(2\pi a \Delta)}{(2\pi \Delta)^2}
% - \frac{2 \sin(2\pi a \Delta)}{\pi \Delta^3} \nonumber
% \end{align}

{\small
\begin{align}
&\int_{-a}^{a} f^{2} e^{j2\pi f\Delta}\,df
= -\frac{1}{(2\pi)^{2}} \, \frac{d^{2}}{d\Delta^{2}}
   \left(\frac{\sin(2\pi a \Delta)}{\pi \Delta}\right) \\ \nonumber
&= \frac{\pi^{2}a^{2}\Delta^{2}\sin(2\pi a\Delta)
   + \pi a\Delta\cos(2\pi a\Delta)
   - \tfrac{1}{2}\sin(2\pi a\Delta)}
   {\pi^{3}\Delta^{3}}.\nonumber
\end{align}
}
Now substituting $a = \frac{B}{2}$ we get the desired result for the FIM in \eqref{eq_FIM_tau}.

\subsection{Generality of the $(S_n,r_n)$ Reformulation}
\label{appendix_generality_Sn_rn}

The reformulation in Section~\ref{subsection_Sn_rn_reformulation} is not specific to the proposed diffraction-path model. It applies to any 2D localization modality whose FIM for target \(n\) can be written as
\[
\bm{\mathcal I}_{n}^{\langle 2\mathrm D\rangle}
=
\sum_{k} \lambda_{k,n}\,\bm g_{k,n}\bm g_{k,n}^{T},
\qquad
\|\bm g_{k,n}\|=1,
\]
where each term $k$ corresponds to one measurement contribution.
Refer to \cite{zekavat2019handbook} for the standard FIM definitions for Euclidean TOA ranging, RSSI and AOA-based localization. Since \(\|\bm g_{k,n}\|=1\), there exists an angle \(\psi_{k,n}\) such that
\[
\bm g_{k,n}
=
\begin{bmatrix}
\cos\psi_{k,n}\\
\sin\psi_{k,n}
\end{bmatrix}.
\]
Hence, defining $S_n,r_n,u_n,v_n$ as in \eqref{eq_definition_perimeter}, the FIM admits the same form as in Section~\ref{subsection_Sn_rn_reformulation}:
\begin{equation}
\bm{\mathcal I}_{n}^{\langle 2\mathrm D\rangle}
=
\frac{1}{2}
\begin{bmatrix}
S_n+u_n & v_n\\
v_n & S_n-u_n
\end{bmatrix}.
\label{eq_appendix_generic_Snr_form}
\end{equation}

To illustrate this, let
\[
\bm p_n \triangleq
\begin{bmatrix}
x_n\\y_n
\end{bmatrix},
\qquad
\bm a_k \triangleq
\begin{bmatrix}
a_{x,k}\\a_{y,k}
\end{bmatrix},
\]
and define
\[
d_{k,n}\triangleq \|\bm p_n-\bm a_k\|,
\qquad
\phi_{k,n}\triangleq \angle(\bm p_n-\bm a_k).
\]
Here, \(d_{k,n}\) is the Euclidean distance between anchor \(k\) and target \(n\), and \(\phi_{k,n}\) is the corresponding bearing angle. Further, \(\sigma_{k,n}^2\) denotes the variance of the measurement error associated with the \(k\)th anchor-target link for the corresponding modality, and \(\beta_k\) denotes the path-loss exponent in the RSSI model. Now, \(\psi_{k,n}\) reduces to the bearing angle \(\phi_{k,n}\) for Euclidean TOA ranging and RSSI, while for AOA it becomes the tangential angle \(\phi_{k,n}+\pi/2\) (modulo \(\pi\)). This is summarized in Table~\ref{tab_modalities_Snr}. 

\begin{table}[t]
\centering
\caption{Examples of 2D localization modalities that admit the \((S_n,r_n)\) reformulation.}
\label{tab_modalities_Snr}
\renewcommand{\arraystretch}{1.15}
\begin{tabular}{p{2.0cm}p{1.5cm}p{1.5cm}p{1.5cm}}
\hline
Modality & \(\bm g_{k,n}\) & \(\lambda_{k,n}\) & \(\psi_{k,n}\) \\
\hline
Euclidean TOA / ranging
&
\(
\begin{bmatrix}
\cos\phi_{k,n}\\
\sin\phi_{k,n}
\end{bmatrix}
\)
&
\(
\dfrac{1}{\sigma_{k,n}^2}
\)
&
\(
\phi_{k,n}
\)
\\[1.2ex]

RSSI
&
\(
\begin{bmatrix}
\cos\phi_{k,n}\\
\sin\phi_{k,n}
\end{bmatrix}
\)
&
\(
\dfrac{\beta_k^2}{\sigma_{k,n}^2 d_{k,n}^2}
\)
&
\(
\phi_{k,n}
\)
\\[1.2ex]

AOA
&
\(
\begin{bmatrix}
-\sin\phi_{k,n}\\
\cos\phi_{k,n}
\end{bmatrix}
\)
&
\(
\dfrac{1}{\sigma_{k,n}^2 d_{k,n}^2}
\)
&
\(
\phi_{k,n}+\dfrac{\pi}{2}
\)
\\
\hline
\end{tabular}
\end{table}

Thus, the \((S_n,r_n)\) reformulation and the associated polygon interpretation are not limited to the proposed diffraction-path ranging model. Extensions to TDOA and other modalities are also possible, but are beyond the scope of this paper.

\bibliographystyle{IEEEtran}{
\footnotesize
\bibliography{refs}
}
\end{document}

%% file: notation.tex
\def\nb0{{\mathbf{0}}}
\def\nb1{{\mathbf{1}}}

\newtheorem{lemma}{Lemma}
\newtheorem{thm}{Theorem}

\newtheorem{ndef}{Definition}

\newtheorem{remark}{Remark}

\def\figref#1{Fig.\,\ref{#1}}%
\DeclareMathOperator{\rank}{rank}